\documentclass[runningheads,a4paper]{llncs}

\usepackage[
margin=2.84cm,
includefoot,
footskip=30pt,
]{geometry}

\usepackage{sgame}
\usepackage{graphicx}
\usepackage{tcolorbox}
\usepackage{amsmath}
\usepackage{amsfonts}
\usepackage{verbatim}
\usepackage{bbm}
\usepackage{amssymb}

\usepackage{mathtools}
\usepackage{cases}
\usepackage{float}
\usepackage{color,soul}
\usepackage{tablefootnote}

\newcommand{\p}{\ensuremath{\mathtt{P}}\xspace}

\usepackage{array}
\usepackage{multirow}

\usepackage{amssymb}
\usepackage{graphicx}

\usepackage{fancybox}
\usepackage{minibox}
\usepackage[export]{adjustbox}
\usepackage{ragged2e}
\usepackage{sgame}
\usepackage{graphicx}
\usepackage{tcolorbox}
\usepackage{amsmath}
\usepackage{amsfonts}
\usepackage{amsthm}
\usepackage{verbatim}
\usepackage{bbm}

\usepackage{mathtools}
\usepackage{cases}
\usepackage{color,soul}
\usepackage{esvect}
\usepackage[nospace,noadjust]{cite}
\usepackage{todonotes}
\usepackage{pgfplots}
\pgfplotsset{width=14cm,compat=1.9}
\usepackage{float}
\usepackage{wrapfig}
\usepackage{enumerate}
\usepackage{xspace}
\usepackage{enumitem}
\usepackage[english]{babel}
\usepackage{algorithm}
\usepackage[noend]{algorithmic}

\newcommand{\NP}{\ensuremath{\mathtt{NP}}\xspace}

\makeatother

\usepackage{amsmath}
\usepackage{amssymb}
\usepackage{xspace}
\usepackage{esvect}

\usepackage{float}
\usepackage{graphicx}
\usepackage{tcolorbox}
\usepackage{wrapfig}

\usepackage[T1]{fontenc}

\usepackage{color,soul}
\usepackage{xcolor}

\DeclareMathOperator{\conv}{conv}

\DeclareMathOperator{\exact}{exact}
\DeclareMathOperator{\stm}{STM}
\DeclareMathOperator{\tmv}{TMV}
\DeclareMathOperator{\tml}{TML}
\DeclareMathOperator{\tsd}{TSD}

\newcommand{\reals}{\mathbb{R}}

\newcommand{\eps}{\ensuremath{\epsilon}}
\newcommand{\leps}{\ensuremath{<_\eps}}
\newcommand{\leeps}{\ensuremath{\le_\eps}}

\newcommand{\geps}{\ensuremath{>_\eps}}
\newcommand{\geeps}{\ensuremath{\ge_\eps}}
\newcommand{\eetr}{\ensuremath{\eps}\text{-ETR}\xspace}
\newcommand{\poly}{\text{poly}\xspace}

\newcommand{\pr}{\ensuremath{\text{Pr}}\xspace}

\newcommand{\xcal}{\ensuremath{\mathcal{X}}}
\newcommand{\ycal}{\ensuremath{\mathcal{Y}}\xspace}
\newcommand{\scal}{\ensuremath{\mathcal{S}}}

\newcommand{\ebb}{\ensuremath{\mathbb{E}}}

\newcommand{\xhat}{\ensuremath{\hat{x}}}

\newcommand{\TFNP}{\ensuremath{\mathtt{TFNP}}\xspace}

\newcommand{\etr}{\ensuremath{\mathtt{ETR}}\xspace}
\newcommand{\fetr}{\ensuremath{\mathtt{FETR}}\xspace}
\newcommand{\tfetr}{\ensuremath{\mathtt{TFETR}}\xspace}

\newcommand{\FNP}{\ensuremath{\mathtt{FNP}}\xspace}
\newcommand{\PSPACE}{\ensuremath{\mathtt{PSPACE}}\xspace}

\newcommand{\CH}{\ensuremath{\textsc{Consensus Halving}}\xspace}

\newcommand{\feas}{\ensuremath{\textsc{Feasible}}\xspace}

\DeclareMathOperator{\gap}{gap}

\usepackage{amsmath}
\usepackage{amsfonts}
\usepackage{array}
\usepackage{multirow}

\usepackage{amsmath}
\usepackage{amssymb}
\usepackage{xspace}
\usepackage{esvect}

\usepackage{color,soul}

\usepackage{float}
\newcommand{\polylog}{\text{poly}\ensuremath{\log}\xspace}

\usepackage{amssymb,amsmath,amsthm}
\usepackage{hyperref} % It is often advices to load hyperref last (there are exceptions).
\newcount\Comments  % 0 suppresses notes to selves in text
\definecolor{darkgreen}{rgb}{0,0.6,0}
\newcommand{\kibitz}[2]{\ifnum\Comments=1{\color{#1}{#2}}\fi}

\begin{document}

\title{Approximating the Existential Theory of the Reals}

%\titlerunning{Approximating the Existential Theory of the Reals}

%\date{}

\author{
	Argyrios Deligkas\inst{1}\thanks{This work was completed while the author was affiliated with the Department of Computer Science and the Leverhulme Research Centre for Functional Materials Design at the University of Liverpool, Liverpool, UK. The author was supported by Leverhulme Trust and the Leverhulme Research Centre for Functional Materials Design.}
	\and
	John Fearnley\inst{2}\thanks{The author was supported by the EPSRC grant EP/P020909/1.}
	\and
	Themistoklis Melissourgos\inst{3}\thanks{This work was completed while the author was affiliated with the Department of Computer Science at the University of Liverpool, Liverpool, UK.}
	\and 
	Paul G. Spirakis\inst{2,4}\thanks{The author was supported by the EPSRC Grant EP/P02002X/1, the Leverhulme Research Centre, and the ERC Project ALGAME.}
}
\institute{Royal Holloway University of London, London, UK\\ \email{Argyrios.Deligkas@rhul.ac.uk} \and Department of Computer Science, University of Liverpool, Liverpool, UK\\
\email{\{John.Fearnley,P.Spirakis\}@liverpool.ac.uk}  \and Operations Research Group, Technical University of Munich, Munich, Germany\\
\email{Themistoklis.Melissourgos@tum.de} \and Computer Engineering and Informatics Department, University of Patras, Patras, Greece}

\authorrunning{A. Deligkas, J. Fearnley, T. Melissourgos, and P. G. Spirakis}

%\maketitle

%REMOVE TITLE FROM TABLE OF CONTENTS
{\def\addcontentsline#1#2#3{}\maketitle}
%

%FORCE TABLE OF CONTENTS BE AT THE SAME PAGE AS ABSTRACT
%\begin{minipage}{\textwidth}
	
	\begin{abstract}
		
		The Existential Theory of the Reals (ETR) consists of existentially quantified Boolean
		formulas over equalities and inequalities of polynomial functions of variables in $\reals$.
		In this paper we propose and study the approximate existential theory of the reals ($\eps$-ETR), in which the
		constraints only need to be satisfied approximately. We first show that when the domain of the variables is $\reals$ then $\eps$-ETR = ETR under polynomial time reductions, and then study the constrained $\eps$-ETR problem when the variables are constrained to lie in a given bounded convex set.
		Our main theorem is a sampling theorem, similar to those that have been proved
		for approximate equilibria in normal form games. It discretizes the domain in a grid-like manner whose density depends on various properties of the formula. A consequence of our
		theorem is that we obtain a quasi-polynomial time approximation scheme (QPTAS) for a
		fragment of constrained $\eps$-ETR. We use our theorem to create several new
		PTAS and QPTAS algorithms for problems from a variety of fields.
		
	\end{abstract}

	\newpage

	\tableofcontents
	
%\end{minipage}

~\\
~\\

\section{Introduction}

\subsection{Sampling techniques}

The Lipton-Markakis-Mehta algorithm (LMM) is a well known method for computing
approximate Nash equilibria in normal form games \cite{LMM}. The key idea behind their technique is to prove that there
exist approximate Nash equilibria where all players use \emph{simple}
strategies. 

Suppose that we have a convex set $C = \conv(c_1, c_2, \dots, c_l)$ defined by
vectors $c_1$ through $c_l$. A vector $x \in C$ is \emph{$k$-uniform} if it can
be written as a sum of the form
%\begin{equation*}
$(\beta_1/k) \cdot c_1 + (\beta_2/k) \cdot c_2 + \dots + 
(\beta_l/k) \cdot c_l,$
%\end{equation*}
where each $\beta_i$ is a non-negative integer and $\sum_{i=1}^l \beta_i = k$.
Since there are at most $l^{O(k)}$ $k$-uniform vectors, we can enumerate all
$k$-uniform vectors in $l^{O(k)}$ time. For approximate equilibria in $n \times n$ bimatrix
games, Lipton, Markakis, and Mehta showed that for every $\eps>0$ there exists 
an $\eps$-Nash
equilibrium where both players use $k$-uniform strategies where $k \in O(\log
n/\eps^2)$, and so they obtained a quasi-polynomial time approximation scheme (QPTAS)
for finding an $\eps$-Nash equilibrium.

Their proof of this fact uses a sampling argument. Every bimatrix game has an
exact Nash equilibrium (NE), and each player's strategy in this NE is a probability
distribution. If we sample from each of these distributions $k$ times, and then
construct new $k$-uniform strategies using these samples, then when $k \in
O(\log n / \eps^2)$ there is a positive probability the new strategies form an
$\eps$-NE. So by the probabilistic method, there must exist a
$k$-uniform $\eps$-NE.

The sampling technique has been widely applied. It was initially used by
Alth\"ofer~\cite{A94} in zero-sum games, before being applied to non-zero sum
games by Lipton, Markakis, and Mehta~\cite{LMM}. Subsequently, it was used to
produce algorithms for finding approximate equilibria in normal form games with
many players~\cite{BBP}, sparse bimatrix games~\cite{Barman15}, tree
polymatrix~\cite{BLP15}, and Lipschitz games~\cite{DFS-lipschitz}.
It has also been used to find 
constrained approximate equilibria in polymatrix games with
bounded treewidth~\cite{DFS17}.

At their core, each of these results uses the sampling technique in the same way
as the LMM algorithm: first take an exact solution to the problem, then sample
from this solution $k$ times, and finally prove that with positive probability
the sampled vector is an approximate solution to the problem. The details of the
proofs, and the value of $k$, are often tailored to the specific application, but
the underlying technique is the same.

\subsection{The existential theory of the reals}

In this paper we ask the following question: \emph{is there a broader class of
	problems to which the sampling technique can be applied?} We answer this by
providing a sampling theorem for the existential theory of the reals.
The existential theory of the reals consists of existentially quantified
formulae using the connectives $\{\land, \lor, \lnot\}$ over polynomials
compared with the operators $\{<, \le, =, \ge, >\}$. For example, each of the
following is a formula in the existential theory of the reals.
\begin{align*}
&\exists x \exists y \exists z \cdot (x = y) \land (x > z) &
&\exists x \cdot (x^2 = 2) \\
&\exists x \exists y \cdot \lnot (x^{10} = y^{100}) \lor (y \ge 4) & 
&\exists x \exists y \exists z \cdot (x^2 + y^2 = z^2)
\end{align*}
Given a formula in the existential theory of the reals, we must decide whether
the formula is \emph{true}, that is, whether there do indeed exist values for the variables
that satisfy the formula. Throughout this paper we will use the Turing model of computation (also known as bit model). In this model, the inputs of our problems will be polynomial functions represented by tensors with rational entries which are encoded as a string of binary bits.

\etr is defined as the class that contains every problem that can be reduced in polynomial time to the typical \etr problem: Given a Boolean formula $F$, decide whether $F$ is a true sentence in the existential theory of the reals. It is known that in the Turing model $\etr \subseteq \PSPACE$ \cite{C88}, and $\NP \subseteq \etr$ since the problem can easily encode
Boolean satisfiability. However, the class is not known to be equal to either
\PSPACE or \NP, and it seems to be a distinct class of problems between
the two. Many problems are now known to be \etr-complete, including various
problems involving constrained equilibria in normal form games with at least
three players~\cite{BM12,BM14,BM16,BM17,GMV+}.

\subsection{Our contribution}

In this paper we propose the \emph{approximate} existential theory of the reals
(\eetr), where we seek a solution that approximately satisfies the
constraints of the formula. We show a subsampling theorem for a large fragment
of \eetr, which can be used to obtain PTASs and QPTASs for the problems
that lie within it. We believe that this will be useful for future research:
instead of laboriously reproving subsampling results for specific games, it now
suffices to simply write a formula in \eetr and then apply our theorem to immediately get the desired result. To exemplify this, we prove several new QPTAS and PTAS
results using our theorem.

Our first result is actually that, in the computational complexity world, $\eps$-\etr = \etr, meaning that the problem of computing an approximate solution to an ETR formula is as hard as finding an exact
solution. However, this result crucially relies on the fact that ETR formulas
can have solutions that are doubly-exponentially large. This motivates the study
of \emph{constrained} $\eps$-ETR, where the solutions are required to lie within
a given bounded convex set. 

Our main theorem (Theorem~\ref{thm:main}) gives a subsampling result for
constrained $\eps$-\etr. It states that if the formula has an exact solution,
then it also has a $k$-uniform approximate solution, where the value of $k$
depends on various parameters of the formula, such as the number of constraints
and the number of variables. The theorem allows for the formula to be written
using \emph{tensor} constraints, which are a type of constraint that
is useful in formulating game-theoretic problems.

The consequence of the main theorem is that, when various parameters of the
formula are up to polylogarithmic in specific parameters (see Corollary~\ref{cor:qptas}), we are able to obtain a
QPTAS for approximating the existential theory of the reals. Specifically, this
algorithm either finds an approximate solution of the constraints, or verifies
that no exact solution exists. In many game theoretic and fair division applications an exact
solution always exists, and so this algorithm will always find an approximate
solution.

We should mention here also that our technique allows approximation of optimization problems whose objective function does not need to be described using the grammar of ETR formulas. For a discussion on this, see Remark \ref{rmk:beyond_ETR}. Also, we are not just applying the well-known subsampling
techniques in order to derive our main theorem. Our main theorem incorporates a
new method for dealing with polynomials of degree $d$, which prior subsampling
techniques were not able to deal with. 

Our theorem can be applied to a wide variety of problems. In the game theoretic
setting, we prove new results for constrained approximate equilibria in normal
form games, and approximating the value vector of a Shapley game. Then we move to the fair division setting, and we show how a special case of the Consensus Halving problem admits a QPTAS. We also show
optimization results. Specifically, we give approximation algorithms for
optimizing polynomial functions over a bounded convex set, subject to polynomial
constraints. We also give algorithms for approximating eigenvalues and
eigenvectors of tensors. Finally, we apply our results to some problems from
computational geometry.

\section{The Existential Theory of the Reals}

Let $x_1$, $x_2$, $\dots$, $x_q \in \reals$ be distinct variables, which we will
treat as a vector $x \in \reals^q$. A \emph{term} of a multivariate polynomial is a function
$T(x) := a \cdot x_1^{d_1} \cdot x_2^{d_2} \cdot \dots \cdot x_q^{d_q}$, where $a$ is non negative rational and $d_1, d_2,
\dots, d_q$ are non negative integers. A multivariate polynomial is a function $p(x) := T_1(x)
+ T_2(x) + \dots + T_t(x) + c$, where each $T_i$ is a term as defined above, and
$c \in \mathbb{Q}_{\geq 0}$ is a constant.

We now define \emph{Boolean formulae} over multivariate polynomials. The atoms
of the formula are polynomials compared with $\{<, \le, =, \ge, >\}$, and the
formula itself can use the connectives $\{\land, \lor, \lnot\}$.

\begin{definition}
	The {\em existential theory of the reals} consists of every true sentence of the form
	$\exists x_1 \exists x_2 \dots \exists x_q \cdot F(x)$,
	where $F$ is a Boolean formula over multivariate polynomials of
	$x_1$ through $x_q$.
\end{definition}
\etr is defined as the class that contains every problem that can be reduced in polynomial time to the typical \etr problem: Given a Boolean formula $F$, decide whether $F$ is a true
sentence in the existential theory of the reals. We will say that $F$ has $m$
constraints if it uses $m$ operators from the set $\{<, \le, =, \ge, >\}$ in its
definition.

\subsection{The approximate \etr}

In the \emph{approximate} existential theory of the reals, we replace the
operators $\{<, \le, \ge, >\}$ with their approximate counterparts. We define
the operators $<_\eps$ and $>_\eps$ with the interpretation that $x <_\eps y$
holds if and only if $x < y + \eps$ and $x >_\eps y$ if and only if $x > y
- \eps$ for some given $\eps > 0$. The operators $\le_\eps$ and $\ge_\eps$ are defined analogously. 

We do not allow equality tests in the approximate \etr. Instead, we require that
every constraint of the form $x = y$ should be translated to $(x \le y) \land (y
\le x)$ before being weakened to $(x \le_\eps y) \land (y \le_\eps x)$.

We also do not allow negation in Boolean formulas. Instead, we require that all
negations are first pushed to atoms, using De Morgan's laws, and then further
pushed into the atoms by changing the inequalities. So the formula $\lnot ((x
\le y) \land (a > b))$ would first be translated to $(x > y) \lor (a \le b)$
before then being weakened to $(x >_\eps y) \lor (a \le_\eps y)$.

%Note that the combination of these two constraints has an interesting effect
%upon the operator $x \ne y$ in the resulting logic. We would first translate
%this to $\lnot ((x \ge y) \land (x \le y))$ and then to $(x \le y) \lor (x \ge
%y)$ and finally to $(x \le_\eps y) \lor (x \ge_\eps y)$. Note that this
%translation is satisfied when $x = y$, while the natural interpretation of $x
%\ne_\eps y$ would be that $x$ and $y$ should differ by at least $\eps$.

\begin{definition}
	The approximate existential theory of the reals consists of every true
	sentence of the form $\exists x_1 \exists x_2 \dots \exists x_q \cdot F(x)$,
	where $F$ is a negation-free Boolean formula using the operators $\{<_\eps,
	\le_\eps, \ge_\eps, >_\eps\}$ over multivariate polynomials of $x_1$ through
	$x_q$.
\end{definition}

Given a Boolean formula $F$, the $\eps$-\etr problem asks us to decide whether
$F$ is a true sentence in the approximate existential theory of the reals, where
the operators $\{<_\eps, \le_\eps, \ge_\eps, >_\eps\}$ are used.

\subsubsection{Unconstrained $\eps$-\etr}

Our first result is that if no constraints are placed on the value of the
variables, that is if each $x_i$ can be arbitrarily large, then $\eps$-\etr = \etr
for \emph{all} values of $\eps > 0$. We show this via a two way polynomial time reduction
between $\eps$-\etr and \etr. The reduction from $\eps$-\etr to \etr is trivial,
since we can just rewrite each constraint $x <_\eps y$ as $x < y + \eps$, and
likewise for the other operators.

For the other direction, we show that the \etr-complete problem $\feas$, which
asks us to decide whether a system of multivariate polynomials
$(p_i)_{i=1,\dots, k}$ has a shared root, can be formulated in $\eps$-\etr. 
%
%%%%%%%%%%%%%%%%%%%%%%%%%%%%%%%%%%%%%%%%%%%%
%
We will prove this by modifying a technique of Schaefer
and Stefankovic~\cite{SS17}.

\begin{definition}[$\feas$]
	Given a system of $k$ multi-variate polynomials $p_i : \reals^{n} \rightarrow
	\reals^{n}$, $i = 1, \dots, k$, decide whether there exists an $x \in \reals^n$ such that $p_i(x) =
	0$ for all $i$.
\end{definition}

Schaefer and Stefankovic showed that this problem is \etr-complete.

\begin{theorem}[\cite{SS17}]
	$\feas$ is \etr-complete.
\end{theorem}

We will reduce $\feas$ to $\eps$-\etr.
%Schaefer and Stefankovic~\cite{SS17} showed that $\feas$ can be reduced to \etr
%in which only strict inequalities are permitted, and our proof will modify their
%technique to produce an \etr instance in which only $<_\eps$ is permitted.
Let $P = (p_i)_{i=1,\dots, k}$ be an instance of $\feas$, and let $L$ be the
number of bits needed to represent this instance. We define $\gap(P) =
2^{-2^{L+5}}$. The following lemma was shown by  Schaefer and Stefankovic. 

\begin{lemma}[\cite{SS17}]
	\label{lem:twotwol}
	Let $P = (p_i)_{i=1,\dots, k}$ be an instance of $\feas$. If there does not exist an
	$x \in \reals^n$ such that $p_i(x) = 0$ for all $i$, then for every $x 
	\in \reals^n$ there exists an $i$ such that $| p_i(x) | > \gap(P)$.
\end{lemma}

In other words, if the instance of $\feas$ is not solvable, then one of the
polynomials will always be bounded away from $0$ by at least $\gap(P)$.

\paragraph*{The reduction}
The first task is to build an $\eps$-\etr formula that ensures that a variable $t \in \reals$
satisfies $t \ge \eps/\gap(P)$. This can be done by the standard trick of
repeated squaring, but we must ensure that the $\eps$-inequalities do not
interfere with the process. We define the following formula over the variables
$t, g_1, g_2, \dots, g_{L+6} \in \reals^n$, where all of the following
constraints are required to hold.
\begin{align*}
g_1 &\ge_{\eps} 2+\eps, \\
g_j &\ge_{\eps} g_{j-1}^2 + \eps & \text{for all $j \in \{2, 3, \dots, L+6\}$.} \\
t & \ge_{\eps} \eps \cdot g_{L+6}  + \eps
\end{align*}
In other words, this requires that $g_1 \ge 2$, and $g_j \ge g_{j-1}^2$. So we
have $g_{L+6} \ge 2^{2^{L+5}}$, and hence $t \ge \eps / \gap(P)$. Note that the
size of this formula is polynomial in the size of $P$.

Given an instance $P = (p_i)_{i=1,\dots, k}$ of $\feas$ we create the following
$\eps$-\etr instance $\psi$, where all of the following are required to hold.
\begin{align}
\label{l1}
t \cdot p_i(x) &\le_{\eps} 0 & \text{for all $i$,}  \\
\label{l2}
t \cdot p_i(x) &\ge_{\eps} 0 & \text{for all $i$,}  \\
\label{l3}
t &\ge \eps / \gap(P),
\end{align}
where the final inequality is implemented using the construction given above.

\begin{lemma}
	$\psi$ is satisfiable if and only if $P$ has a solution.
\end{lemma}
\begin{proof}
	First, let us assume that $P$ has a solution.
	This means that there exists an $x \in \reals^n$ such that $p_i(x) = 0$ for all
	$i$. Note that 
	$x$ clearly satisfies inequalities \eqref{l1} and \eqref{l2}, while
	inequality \eqref{l3} can be satisfied by fixing $t$ to be any number greater than
	$\eps/\gap(P)$. So we have proved that $\psi$ is satisfiable.
	
	On the other hand, now we will assume that $x \in \reals^n$ satisfies $\psi$. 
	Note that we must have 
	\begin{equation*}
	p_i(x) \le \eps / t \le \gap(P)
	\end{equation*}
	and likewise 
	\begin{equation*}
	p_i(x) \ge -\eps / t \ge -\gap(P),
	\end{equation*}
	and hence $|p_i(x)| \le \gap(P)$ for all $i$. But Lemma~\ref{lem:twotwol} states
	that this is only possible in the case where $P$ has a solution.
\end{proof}
This completes the proof of the following theorem.

%%%%%%%%%%%%%%%%%%%%%%%%%%%%%%%%%%%%%%%%%%%

\begin{theorem}
	\label{thm:unconstrained}
	$\eps$-\etr = \etr for all $\eps \geq 0$.
\end{theorem}

\subsubsection{Constrained $\eps$-\etr} 

In our negative result for unconstrained $\eps$-\etr, we abused the fact that
variables could be arbitrarily large to construct the doubly-exponentially large
number~$t$. So, it makes sense to ask whether $\eps$-\etr gets easier if we
\emph{constrain} the problem so that variables cannot be arbitrarily large.

%\john{Shouldn't F be defined to be an $\eps$-\etr formula, rather than just any
%Boolean formula?}
In this paper, we consider $\eps$-\etr problems that are constrained by a bounded convex
set in $\reals^q$. For vectors $c_1, c_2, \dots, c_l \in \reals^q$  we use
$\conv(c_1, c_2, \dots, c_l)$ to denote the set containing every vector that
lies in the convex hull of $c_1$ through $c_l$. In the {\em constrained $\eps$-\etr},
we require that the solution of the $\eps$-\etr problem should also lie in the
convex hull of $c_1$ through $c_l$.

\begin{definition}\label{def:const_eps_ETR}
	Given vectors $c_1, c_2, \dots, c_l \in \reals^q$ and a Boolean formula $F$ that uses the operators $\{<_\eps,
	\le_\eps, \ge_\eps, >_\eps\}$, the constrained $\eps$-\etr problem asks us to decide whether
	\begin{equation*}
	\exists x_1 \exists x_2 \dots \exists x_q \cdot \bigl(x \in \conv(c_1, c_2, \dots,
	c_l) \land F(x) \bigr).
	\end{equation*}
\end{definition}

Note that, unlike the constraints used in $F$, the convex hull constraints are
not weakened. So the resulting solution $x_1, x_2$, $\dots$, $x_q$, must
actually lie in the convex set.

\section{Approximating Constrained $\eps$-\etr}

\subsection{Polynomial classes}

To state our main theorem, we will use a certain class of polynomials where the
coefficients are given as a tensor. This will be particularly useful when we
apply our theorem to certain problems, such as normal form games. To be clear
though, this is not a further restriction on the constrained $\eps$-\etr problem,
since all polynomials can be written down in this form. 

In the sequel, we use the term {\em variable} to refer to a $p$-dimensional vector; for example, in Definition \ref{def:const_eps_ETR}, the $q$-dimensional vector $x$ would be called variable under this new naming. The variables of the polynomials we will study will be $p$-dimensional vectors
denoted as $x_1, x_2, \ldots, x_n$, where $x_j (i)$ will denote the $i$-th
element ($i \in [p]$) of vector $x_j$. The coefficients of the polynomials will be a captured by tensor
denoted by $A$. Given a $\times_{j=1}^{n} p$ tensor $A$, we denote by $a(i_1,
\ldots, i_n)$ its element with coordinates $(i_1, \ldots, i_n)$ on the tensor's
dimensions $1, \ldots, n$, respectively, and by $\alpha$ we denote the maximum
absolute value of these elements. We define the following two classes of
polynomials.

\begin{itemize}
	\item {\bf Simple tensor multivariate.}  
	%In a simple tensor multivariate polynomial (STM), every variable $x_j$ has a \emph{maximum
	%degree} given by $d_j$. Every possible term $x_1^{i_1} \cdot x_2^{i_2} \cdot
	%\dots \cdot x_m^{i_m}$ appears in the polynomial, as long as each $i_j \le d_j$,
	%and the coefficient of this term is given by a tensor. 
	
	We will use $\stm(A,x^{d_1}_{1},\ldots,x^{d_n}_{n})$ 
	%\themis{The number of variables should, in general, not depend on their dimension (here they are both $n$). I remind that the DIMENSION of the variables is the actual input size $n$ (it would correspond to the number of strategies in an \etr formulation of a multi-player game, while the NUMBER of the variables would be the number of distinct types of players, and $\sum_{j} d_j$ would be the number of players). The number of variables depends on the application. Let me denote the number of variables by $m$, and their dimension $n$. Consider also the notion of the degree $d$ of the TMV definition below. It is obvious that a TMV with degree $d$ has number of STMs $t$ that is TIGHTLY upper bounded by $d^m$. Back to the multi-player game, this means that for a constant number of players $\sum_{j \in [m]} d_j$, it will be $d=$const, and also the number of distict types $m$ is constant. Then $t \leq d^m$ will be constant}
	denote an STM polynomial with $n$ variables where each variable $x_j$, $j \in [n]$ is applied $d_j$ times on tensor $A$ that defines the coefficients. Tensor $A$ has $\sum_{j=1}^{n} d_j$ dimensions with $p$ indices each. We will say that an STM polynomial is of maximum 
	degree $d$, if $d = \max_{j}d_j$.
	Here is an example of a degree 2 simple tensor polynomial with two variables:
	\begin{align*}
	\stm(A, x^{2},y) = \sum_{i=1}^p \sum_{j=1}^p  \sum_{k=1}^p x(i)\cdot x(j) \cdot y(k) \cdot a(i,j,k) + 10.
	\end{align*}
	
	%and can be partially symmetric, i.e. symmetric only in the dimensions on which the same vectors are applied, similarly to the standard degree $d$ case. 
	This polynomial itself is written as follows.
	\begin{align*}
	&\stm(A, x^{d_1}_1,\ldots,x^{d_n}_n) = \\
	&\sum_{i_{1,1} \in [p]}\cdots\sum_{i_{n,d_n} \in [p]}
	(x_1(i_{1,1}))\cdot\ldots\cdot(x_1(i_{1,d_1}))\cdot\ldots\cdot (x_n(i_{n,1}))\cdot\ldots\cdot(x_n(i_{n,d_n}))\cdot \\
	&\qquad \qquad \qquad \qquad \qquad \qquad \qquad \cdot a(i_{1,1},\ldots,i_{1,d_1} \ldots, i_{n,1}, \ldots, i_{n,d_n}) + a_0.
	\end{align*}
	
	\item {\bf Tensor multivariate.} 
	A tensor multivariate (TMV) polynomial is the sum over a number of simple tensor
	multivariate polynomials. We will use \text{\quad} $\tmv(x_1,\ldots,x_n)$ to denote a
	tensor multivariate polynomial with $n$ vector variables, which is formally defined as 
	\begin{equation*}
	\tmv(x_1,\ldots,x_n) = \sum_{i \in [t]} \stm(A_i, x^{d_{i1}}_1,\ldots,x^{d_{in}}_n),
	\end{equation*}
	where the exponents $d_{i1},\ldots,d_{in}$ depend on $i$, and $t$ is the number of simple 
	multivariate polynomials.  We will say that $\tmv(x_1,\ldots,x_n)$ has length $t$ if it is the
	sum of $t$ STM polynomials, and that it is of degree $d$ if $d = \max\limits_{i \in [t], j \in [n]}d_{ij}$. Observe that $t \leq (d+1)^n$; it could be the case that a TMV polynomial is a sum of STM polynomials, each of which has a distinct combination of exponents $d_{i1},\ldots,d_{in}$ in its variables, where each $d_{ij} \in \{ 0, 1, \dots, d \}$.
\end{itemize}

\subsection{$\eps$-\etr with tensor constraints}

We focus on $\eps$-\etr instances $F$ where all constraints are of the form
$\tmv(x_1, \ldots, x_n) \bowtie 0 $, where $\bowtie$ is an operator from the set
$\{<_\eps, \le_\eps, >_\eps, \ge_\eps\}$. 
Recall that each TMV constraint considers vector variables. 
We consider the number of variables used in $F$ (denoted as $n$) to be the 
number of vector variables used in the TMV constraints. 
So the value of $n$ used in our main theorem may be constant if
only a constant number of vectors are used, even if the underlying $\eps$-\etr
instance actually has a non-constant number of variables.
For example, if $x$ and $y$ and $w$ are $p$-dimensional probability distributions 
and $A_1$ and $A_2$ are $p \times p$ tensors, the TMV constraint 
$x^TA_1y + w^TA_2x > 0$ has three variables, degree 1, length two; though 
the underlying problem has $3\cdot p$ variables.

Note that every $\eps$-\etr constraint can be written as a TMV constraint,
because all multivariate polynomials can be written down as a TMV polynomial.
Every term of a TMV can be written as a STM polynomial where the tensor entry is non zero for
exactly the combination of variables used in the term, and $0$ otherwise. Then a
TMV polynomial can be constructed by summing over the STM polynomial for each
individual term.

\subsubsection{The main theorem}

Given an $\eps$-\etr formula $F$, we define $\exact(F)$ to be a Boolean formula
in which every approximate constraint is replaced with its exact variant,
meaning that every instance of $x \le_\eps y$ is replaced with $x \le y$, and
likewise for the other operators.

Our main theorem is as follows.

\begin{theorem}
	\label{thm:main}
	Let $F$ be an $\eps$-\etr instance with $n$ vector variables and $m$ multivariate-polynomial
	constraints each one of maximum length $t$ and maximum degree $d$, 
	constrained by the convex hull defined by $c_1, c_2, \ldots, c_l \in \reals^{np}$. 
	Let $\alpha$ be the maximum absolute value of the coefficients of constraints of $F$, 
	and let $\gamma =\max_i \|c_i\|_\infty$. If $\exact(F)$ has a solution in 
	$\conv(c_1, c_2, \ldots, c_l)$, then $F$ has a $k$-uniform solution in 
	$\conv(c_1, c_2, \ldots, c_l)$ where
	%	\begin{align*}
	%	k = \frac{48 \cdot \alpha^6 \cdot \gamma^{2d+2} \cdot d^5 \cdot t^4 \cdot n^6 \cdot 
	%		\ln(2 \cdot \alpha \cdot \gamma \cdot d \cdot t \cdot n \cdot m)}{\eps^4}.
	%	\end{align*}
	\begin{align*}
	k = \frac{512 \cdot \alpha^6 \cdot \gamma^{2d+2} \cdot n^6 \cdot d^6 \cdot t^5 \cdot \ln ( 2 \cdot \alpha' \cdot \gamma' \cdot d \cdot n \cdot m )}{\eps^5},
	\end{align*}
	where $\alpha':=\max(\alpha, 1), \gamma':=\max(\gamma, 1)$.
\end{theorem}

\subsubsection{Consequences of the main theorem}\label{sec:conseq_main}

Our main theorem gives a QPTAS for approximating a fragment of $\eps$-\etr. The
total number of $k$-uniform vectors in a convex set $C = \conv(c_1, c_2, \dots,
c_l)$ is $l^{O(k)}$. So, if the parameters $\alpha$, $\gamma$, $d$, $t$,
and $n$ are all polylogarithmic in $m$, then our main theorem tells us that the total number
of $k$-uniform vectors is $l^{O(\polylog m)}$, where $m$ is the number of
constraints. So if we enumerate each $k$-uniform vector $x$, we can check
whether $F$ holds, and if it does, we can output $x$. If no $k$-uniform vector
exists that satisfies $F$, then we can determine that $\exact(F)$ has no
solution. This gives us the following result.

\begin{corollary}
	\label{cor:qptas}
	Let $F$ be an $\eps$-\etr instance constrained by the convex hull defined by $c_1,
	c_2, \ldots, c_l$. If $\alpha$, $\gamma$, $n$, $d$, and $t$ are polylogarithmic in $m$, then
	we have an algorithm and runs in time $l^{O\left( \frac{\polylog m}{\eps^5} \right)}$  that
	either finds a solution to $F$, or determines that $\exact(F)$ has no solution.
\end{corollary}
Let $N$ be the input size of the given problem. If $m$ is constant and $l$ is polynomial in $N$ then this gives a PTAS, while if $m$ and $l$ are polynomial in $N$, then this gives a QPTAS.

In Section~\ref{sec:appl} we will show that the problem of approximating
the best social welfare achievable by an approximate Nash equilibrium in a
two-player normal form game can be written down as a constrained  $\eps$-\etr
formula where $\alpha$, $\gamma$, $d$, and $n$ are constant (and recall that $t \leq (d+1)^n$). It has been shown
that, assuming the exponential time hypothesis, this problem cannot be solved faster than quasi-polynomial
time~\cite{BKW15,DFS18}, so this also implies that constrained $\eps$-\etr where 
$\alpha$, $\gamma$, $d$, and $n$ are constant cannot be solved faster than
quasi-polynomial time unless the exponential time hypothesis is false.

Many $\eps$-\etr problems are naturally constrained by sets that are defined by
the convex hull of exponentially many vectors. The cube $[0, 1]^p$ is a natural
example of one such set. Brute force enumeration does not give an efficient
algorithm for these problems, since we need to enumerate $l^{O(k)}$ vectors, and
$l$ is already exponential in the dimension parameter $p$. However, our main theorem is able to provide
non-deterministic polynomial time algorithms for these problems. 

This is because
each $k$-uniform vector is, by definition, the convex combination of at most $k$
of the vectors in the convex set, and this holds even if $l$ is exponential. So,
provided that $k$ is polynomial in the input size, we can guess the subset of vectors that are
used, and then verify efficiently that the formula holds. This is particularly useful for
problems where $\exact(F)$ always has a solution, which is often the case in
game theory applications, since it places the
approximation problem in \NP, whereas deciding the existence of an exact solution may be
\etr-complete.

\begin{corollary}
	\label{cor:nptas}
	Let $F$ be an $\eps$-\etr instance constrained by the convex hull defined by $c_1,
	c_2, \ldots, c_l$. If $\alpha$, $\gamma$, $d$, $t$, $n$, are polynomial in the input size,
	then there is a non-deterministic polynomial time algorithm that either finds a
	solution to $F$, or determines that $\exact(F)$ has no solution. Moreover, if
	$\exact(F)$ is guaranteed to have a solution, then the problem of finding an
	approximate solution for $F$ is in \NP. 
\end{corollary}

\subsubsection{Approximation notions}

According to the relaxation procedure for \etr that we have described, each atom $A_i$ of the \etr formula is relaxed additively by a positive quantity $\eps$. The main theorem (Theorem \ref{thm:main}) and the intermediate resluts, give a sufficiently fine discretization (distance at most $1/k$ for some $k \in \mathbb{N}^*$) of the domain of the \etr instance's variables, such that if there exists an exact solution $x^*=(x_1^*, \dots, x_n^*)$ of the formula then there exists a $k$-uniform solution in the discretized domain that $\eps$-satisfies every $A_i$.
In particular we prove that if $A_i = (p(x) \bowtie 0)$, where $p(x)$ is a multivariate polynomial and $\bowtie \in \{ <, \leq, =, \geq, > \}$ then there exists a $k$-uniform vector $x'$ such that $| p(x') - p(x^*) | \leq \eps$. This implies the $\eps$-satisfaction of each $A_i$ by the triangle inequality.

In fact, by this work we do not aim to reply an ``approximate yes/no'' to an \etr instance, i.e. to give a yes/no answer to the relaxed \etr instance, but instead to output an \textit{approximate solution} (if an exact solution exists) to the \etr instance. Therefore, more accurately we should refer to this approximation of \etr as an approximation of \texttt{Function} \etr (\fetr), where \fetr is the function problem extension of the decision problem complexity class \etr. As \etr is the analogue of \NP, \fetr is the analogue of \FNP in the Blum-Shub-Smale computation model \cite{BSS89}.
\begin{definition}[$\eps$-approximation]
	Consider a given \etr instance with domain $D$ and formula $F$. If $x^*$ is a solution to the instance and $x'$ is a solution to the respective $\eetr$ instance for a given $\eps > 0$, then $x'$ is called an \textit{$\eps$-approximation} of $x^*$.
\end{definition}
\begin{definition}[PTAS/QPTAS]
	Consider a function problem $P$ with input size $N$, whose objective is to output a solution $x^*$. An algorithm that computes an $\eps$-approximation $x'$ of $P$ in time polynomial in $N$ for any fixed $\eps>0$ is a Polynomial Time Approximation Scheme (PTAS). An algorithm that computes $x'$ in time $O(N^{\polylog N})$ is a Quasi-Polynomial Time Approximation Scheme (QPTAS).  
\end{definition}

%\john{This is good to point out, but perhaps the arguments for why minmax and
%maxmin can also be done are a bit handwavey. Can we formalise this into a
%theorem?}
\begin{remark}\label{rmk:beyond_ETR}
	Our technique that finds an $x'$ such that $| p(x') - p(x^*) | \leq \eps$ provides one with more power than showing that polynomial inequalities weakened by $\eps$ hold for $x'$. In fact, it allows for approximation of solutions that need not be described by an \etr formula. A simple example of such a case is the one presented in Section \ref{sec:QP-PTAS} where we seek an approximation of the maximum of the quadratic function in the simplex. The maximization objective does not need to be written in an \etr formula. Instead, we show that any point $f(x)$ of the quadratic function, for $x$ in the simplex, can be approximated by a point $f(x')$ where $x'$ is in a discrete simplex with a small number of points. Then we find the maximum of $f(x')'s$ which is smaller than $\max(f(x))$ by at most $\eps$. 
	
	The fact that operation ``max'' can be executed in time linear in the number of points of the discretized simplex allows us to use our method for expressions with ``max'' which is forbidden in the grammar of \etr. More generally, the following theorem shows that even more complicated objectives, such as ``$\max_{x_1}\min_{x_2}$'' can be treated by a modification of the algorithm described in Section \ref{sec:conseq_main}.
\end{remark}

\begin{theorem}\label{thm:multi-obj}
	Let $F$ be a multi-objective optimization instance whose objective functions are multivariate polynomials, with variables constrained by the convex hull defined by $c_1, c_2,$ $\ldots, c_l$. Let $k$ be the quantity specified in Theorem \ref{thm:main} with $m$ being the number of polynomial functions in the instance, meaning the ones in the objectives and constraints. If every objective on the functions has a polynomial time algorithm to be performed on a discrete domain, then there is an algorithm that runs in time $l^{O(k)}$, and either finds a solution which satisfies every objective of $F$ within additive $\eps$, or determines that $F$ has no solution.
\end{theorem}

\begin{proof}
	As explained at the beginning of this section, our technique discretizes the domain of the variables with a density sufficient to approximate any point of any of the polynomial functions that are given as part of the atoms of an \etr formula. That is, for any $x^*$ in the continuous domain it guarantees the existence of a discrete $x'$ such that for every polynomial $p$ in the atoms, it is $| p(x') - p(x^*) | \leq \eps$. Note now that the technique works for any given set of polynomials when we require that for every polynomial in the set, every point $x^*$ has a discrete $x'$. This is regardless of what the atoms' operators from $\{ <, \leq, =, \geq, > \}$ are or with what logical operators from $\{ \land, \lor\}$ the atoms connect to each other.
	
	In view of the above, observe that any objective (with the properties of the statement of the theorem) on functions, takes time polynomial in the size of the discretized space, therefore it does not change asymptotically the total running time of the algorithm described at the beginning of Section \ref{sec:conseq_main}. That is because first, the aforementioned algorithm will brute-force through all of the points in the discretized domain and for these points it will check if all of the constraints of $F$ are satisfied. Now the algorithm we propose will deviate from the aforementioned algorithm and for the points that satisfy the constraints of $F$ (feasible points), for each objective it will run the efficient respective algorithm of the objective on the feasible points and check whether all objectives of the relaxed by $\eps$ instance are satisfied for some point. This can be done in time polynomial in the size of the discretized domain, i.e. $l^{O(k)}$. If a discrete point is found that $\eps$-satisfies $F$, then the algorithm returns it, otherwise there no point in the continuous domain that satisfies $F$ according to Theorem \ref{thm:main}. 
\end{proof}

\subsection{A theorem for non-tensor constraints}

One downside of Theorem~\ref{thm:main} is that it requires that the formula is
written down using tensor constraints. We have argued that every \etr formula can
be written down in this way, but the translation introduces a new
vector-variable for each variable in the \etr formula. When we apply
Theorem~\ref{thm:main} to obtain PTASs or QPTASs we require that the number of
vector variables is at most polylogarithmic, and so this limits the application of the theorem
to \etr formulas that have at most polylogarithmically many variables. 

Theorem \ref{thm:kunif} is a sampling result for $\eps$-\etr with non-tensor
constraints, which is proved via some intermediate results. 
%
%
%%%%%%%%%%%%%%%%%%%%%%%%%%%%%%%%%%%%%%%%%%%%%%%
%
%
First, we will use the following theorem of Barman.
\begin{theorem}[\cite{Barman15}]
	\label{thm:barman}
	Let $c_1, c_2, \dots, c_l \in \reals^q$ with $\max_i \| c_i \|_\infty \le 1$.
	For every $x \in \text{conv}(c_1, c_2,$ $ \dots, c_l)$ and every $\eps > 0$
	there exists a $O(\log l/ \eps^2)$-uniform vector $x' \in \conv(c_1, c_2, \dots, c_l)$ such that $\| x - x' \|_\infty \le \eps$.
\end{theorem}

The following lemma shows that if we take two vectors $x$ and $x'$ that are
close in the $L_\infty$ norm, then for all polynomials $p$ the value of $|p(x) -
p(x')|$ cannot be too large.

We denote by $consts(p)$ the maximum absolute coefficient in polynomial $p$, and by $terms(p)$ the number of terms of $p$.

\begin{lemma}
	\label{lem:polynomial}
	Let $p(x)$ be a multivariate polynomial over $x \in \mathbb{R}^q$ with
	degree $d$ and let $\epsilon \in (0,\gamma]$ for some constant $\gamma >0$. For every pair of vectors $x,x' \in [0,\gamma]^q$ with $\| x - x' \|_{\infty} \leq \epsilon$ we have:
	\begin{align*}
	|p(x) - p(x')| \leq \gamma^{d-1} \cdot (2^d - 1) \cdot consts(p) \cdot terms(p) \cdot \epsilon .
	\end{align*}
\end{lemma}

\begin{proof}
	Consider a term of $p(x)$, which can without loss of generality be written as $t(x) = c \cdot \prod\limits_{\substack{i \in [q] \\ \sum\limits_{i} d_i \leq d}} x_{i}^{d_i}$, where $d_i$ is the degree of coordinate $x_i$ (resp. $x_i'$). We have
	\begin{align*}
	|t(x) - t(x')| & = \left|c \cdot \prod\limits_{\substack{i \in [q] \\ \sum\limits_{i} d_i \leq d}} x_{i}^{d_i} - c \cdot \prod\limits_{\substack{i \in [q] \\ \sum\limits_{i} d_i \leq d}} (x_{i}')^{d_i}\right| \\
	& = c \cdot \left|\prod\limits_{\substack{i \in [q] \\ \sum\limits_{i} d_i \leq d}} x_{i}^{d_i} - \prod\limits_{\substack{i \in [q] \\ \sum\limits_{i} d_i \leq d}} (x_{i}')^{d_i}\right| \\
	& \leq c \cdot \left[ \prod\limits_{\substack{i \in [q] \\ \sum\limits_{i} d_i \leq d}} (x_{i}^{d_i} + \epsilon) - \prod\limits_{\substack{i \in [q] \\ \sum\limits_{i} d_i \leq d}} x_{i}^{d_i} \right] \\
	& \leq c \cdot \left[ \left[ \prod\limits_{\substack{i \in [q] \\ \sum\limits_{i} d_i \leq d}} x_{i}^{d_i} + \binom{d}{1}\gamma^{d-1}\epsilon + \binom{d}{2}\gamma^{d-2}\epsilon^2 + \dots + \binom{d}{d}\gamma^{0}\epsilon^d \right] -  \prod\limits_{\substack{i \in [q] \\ \sum\limits_{i} d_i \leq d}} x_{i}^{d_i} \right] \\
	& \leq c \cdot \epsilon \cdot \sum\limits_{k=1}^{d} \binom{d}{k} \gamma^{d-1}  \\
	& = c \cdot \epsilon \cdot \gamma^{d-1} \cdot \sum\limits_{k=1}^{d} \binom{d}{k} \\
	& = \epsilon \cdot c \cdot \gamma^{d-1} \cdot (2^d - 1),
	\end{align*}
	where the fourth and third to last lines use the fact that $x_i$'s, and $\epsilon$ are all at most $\gamma$.
	
	Next consider a term $t(x)$ of $p(x)$ of degree $d' \leq d$. This can be written similarly to the aforementioned term. Then $|t(x) - t(x')| \leq c \cdot \epsilon \cdot \gamma^{d-1} \cdot (2^{d'} - 1) \leq c \cdot \epsilon \cdot \gamma^{d-1} \cdot (2^{d} - 1)$. Since there are $terms(p)$ many terms in $p$, we therefore have that
	\begin{align*}
	|p(x) - p(x')| \leq \gamma^{d-1} \cdot (2^d - 1) \cdot consts(p) \cdot terms(p) \cdot \epsilon .
	\end{align*}
\end{proof}

%For every $\eetr$ instance $F$, 
%let $\exact(F)$ be a modified
%version of $F$, in which all constraints are replaced with their exact variants
%(so $=_\eps$ is replaced with $=$ and so on). 
%let $\consts(F) = \max_{p}
%\consts(p)$, let $\terms(F) = \max_{p} \terms(p)$, and let $\text{degree}(F) = \max_{p} \text{deg}(p)$, where all three maximums range over the polynomials $p$ used in the constraints of $F$.

We now apply this to prove the following theorem.
\begin{theorem}
	\label{thm:kunif}
	Let $F$ be an $\eetr$ instance constrained over the convex hull defined
	by $c_1, c_2, \dots, c_l \in \reals^q$. 
	Let $m$ be the number of constraints used in $F$, Let $\gamma = \max_i \| c_i \|_{\infty}$, 
	let $\alpha$ be the largest constant coefficient used in $F$, let $t$ be the number of terms
	used in total in all polynomials of $F$, and let $d$ be the
	maximum degree of the polynomials in $F$. 
	If $\exact(F)$ has a solution in $\conv(c_1, c_2, \dots, c_l)$,
	then $F$ has a $k$-uniform solution in 
	$\conv(c_1, c_2, \dots, c_l)$ where 
	\begin{equation*}
	k =  \alpha^2 \cdot \gamma^{2d-2} \cdot (2^d - 1)^2 \cdot t^2 \cdot \log l / \eps^2.
	\end{equation*}
\end{theorem}

\begin{proof}
	Let $x$ be the solution to $\exact(F)$.
	First we apply Theorem~\ref{thm:barman} to find a point $y$ that is $k$-uniform,
	where $k =  \alpha^2 \cdot \gamma^{2d-2} \cdot (2^d - 1)^2 \cdot t^2 \cdot \log l / \eps^2$, such that
	\begin{equation*}
	\| x - y \|_\infty \le \eps/(\alpha \cdot \gamma^{d-1} \cdot (2^d - 1) \cdot t).
	\end{equation*}
	Next we can apply Lemma~\ref{lem:polynomial} to argue that, for each polynomial
	$p$ used in $F$, we have
	\begin{align*}
	| p(x) - p(y) | &\le \alpha \cdot \gamma^{d-1} \cdot (2^d - 1) \cdot t \cdot
	\left(\frac{\eps}{\alpha \cdot \gamma^{d-1} \cdot (2^d - 1) \cdot t}\right) \\
	& = \eps.
	\end{align*}
	Since all constraints of $F$ have a tolerance of $\eps$, and since $x$ satisfies
	$\exact(F)$, we can conclude that $F(y)$ is satisfied.
\end{proof}

%%%%%%%%%%%%%%%%%%%%%%%%%%%%%%%%%%%%%%%%%%%%%%

The key feature here is that the number of variables does not appear in the
formula for $k$, which allows the theorem to be applied to some formulas for
which Theorem~\ref{thm:main} cannot. However, since the theorem does not allow
tensor constraints, its applicability is more limited because the number of
terms $t$ will be much larger in non-tensor formulas. For example, as we will
see in Section~\ref{sec:appl}, we can formulate bimatrix games using tensor
constraints over constantly many vector variables, and this gives a result using
Theorem~\ref{thm:main}. No such result can be obtained via
Theorem~\ref{thm:kunif}, because when we formulate the problem without tensor
constraints, the number of terms $t$ used in the inequalities becomes polynomial in the dimension.

\section{The Proof of the Main Theorem}

In this section we prove Theorem~\ref{thm:main}. Before we proceed with the
technical results, let us illustrate via an example the crucial idea for proving that the special vectors we have defined (i.e. the $k$-uniform vectors for some $k \in \mathbb{N}^*$) inside a discretized convex hull can be used to approximate not only multilinear polynomials, but also multivariate polynomials of degree $d \geq 2$. At the same time, we show that the discretization of the domain (points in distance at most $1/k$ from each other) does not need to be very fine in order to achieve an additive approximation $\eps$ at any point of such a function. Our example is in approximating the quadratic polynomial over the simplex.

Let us provide a roadmap for this section.  We begin by the detailed aforementioned example. Then we proceed by
considering two special cases, namely Lemma \ref{lem:k-TML} and Lemma \ref{lem:k-TSD}, which when combined will be the backbone of the
proof of the main
theorem. 

Firstly, we will show how to deal with problems where every constraint of the
Boolean formula is a \emph{multilinear polynomial}, which we will define
formally later. We deal with this kind of problems using Hoeffding's inequality
and the union bound, which is similar to how such constraints have been handled
in prior work.

Then, we study problems where the Boolean formula
consists of a {\em single} degree $d$ polynomial constraint. We reduce
this kind of problems to a constrained $\eps/2$-\etr problem with
multilinear constraints, so we can use our previous result to handle the reduced
problem. Sampling techniques in degree $d$ polynomial problems have not been considered in previous work, and
so this reduction is a novel extension of sampling based techniques to a broader
class of $\eps$-\etr formulas.

Finally, we deal with the main theorem: we reduce the original \etr
problem with multivariate constraints to a set of $\eps'$-\etr problems with a
single standard degree $d$ constraint, and then we use the last result to derive
a bound on $k$.

As a byproduct of our main result one can get the same result as that of \cite{KLP06} in which a PTAS for fixed degree polynomial minimization over the simplex was presented. Even though the PTAS that follows from our result on the same optimization problem has roughly the same running time as that of \cite{KLP06}, the proof presented here (which is independent of the aforementioned work) is significantly simpler. Nevertheless, the result in the current work generalizes previous results on polynomial optimization over the simplex, by providing a universal algorithm for multi-objective optimization problems, and showing how its running time depends on the parameters of the problem (see Theorem \ref{thm:multi-obj}).

\subsection{Example: A simple PTAS for quadratic polynomial optimization over the simplex}\label{sec:QP-PTAS}

%\paragraph{\bf Probability vectors.}
%
%Consider a discrete probability distribution on the sample space $[n] = \{1,2, \dots, n\}$. We call the vector $x = (x_1, x_2, \dots, x_n)$ a \textit{probability vector} if $x_l \geq 0$ for every $l \in [n]$ and $\sum_{l=1}^{n} = 1$.  The set $S(x)=\{l : x_l > 0\}$ is called the \textit{support} of $x$. A probability vector is called \textit{k-uniform} if it is a uniform distribution on a multiset $M$ of $[n]$, with $|M|=k$. We distinguish between (in general, different) probability vectors by denoting them with corresponding superscripts, e.g. $x^{(i)}, x^{(j)}$. Note that if a matrix $A_{n \times n}$ is given which is also symmetric, i.e. $A = A^T$, then $x^{(i)T} A x^{(j)} = x^{(j)T} A x^{(i)}$ , for any $i,j$.

\begin{definition}[Standard quadratic optimization problem (\textsc{SQP})]
	Given a $p \times p$ matrix $A$ with entries normalized in $[0,1]$, find the value 
	\begin{align*}
	v^*:= \max_{x \in \Delta_{p}}x^{T}Ax, \quad \text{where $\Delta_{p}$ is the $(p-1)-simplex$.}
	\end{align*}
\end{definition}
\textsc{sqp} is a strongly \NP-hard problem, even for the case where $A$ has entries in $\{0,1\}$; in a theorem of Motzkin and Straus \cite{MS65} it is shown that if matrix $A$ is the adjacency matrix of a graph on $p$ vertices whose maximum clique has $c$ vertices, then $v^* = 1 - 1/c$. The problem of finding the size of the maximum clique in a general graph is known to be (strongly) \NP-hard since its decision version is one of Karp's 21 \NP-complete problems \cite{K72}. Therefore, unless \p = \NP there is no Fully Polynomial Time Approximation Scheme for \textsc{sqp} and the best thing we can hope for the problem is a PTAS. We present a PTAS for \textsc{sqp} (Corollary \ref{cor:PTAS_SQP}), which has almost the same running time as that of \cite{BK02}, but we claim that our proof is significantly simpler.

Let $x^* := \arg(v^*)$. Consider the set $\Delta_p(k)$ of all $k$-uniform vectors, for $k= 16 \ln(3/\epsilon)/\epsilon^2$, with items $x^{(i)} \in \Delta_p(k)$, for $i=1,2,\dots, |\Delta_p(k)|$.

\begin{lemma}\label{lem: R}
	There exists a multiset $\xcal$ of $\Delta_p(k)$ with $|\xcal|=2/\epsilon$ such that for every $x^{(i)}, x^{(j)} \in \xcal$ with $i \neq j$, it is
	\begin{align*}
	x^{*T} A x^* - x^{(i)T} A x^{(j)} < \epsilon/2.
	\end{align*}
\end{lemma}

\begin{proof}
	Note that although $i \neq j$, $x^{(i)}$ could be equal to $x^{(j)}$ since the two $k$-uniform vectors belong to a multiset of $\Delta_p(k)$. The proof is by the probabilistic method. Let us create the events 
	\begin{align*}
	E_{i} &= \left\{x^{*T} A x^* - x^{(i)T} A x^*  < \epsilon/4 \right\}, \quad &\forall i \text{ for which } x^{(i)} \in \xcal, \\
	F_{i,j} &= \left\{x^{(i)T} A x^* - x^{(i)T} A x^{(j)} < \epsilon/4\right\}, \quad &\forall i,j \text{ with } i \neq j, \text{ for which } x^{(i)},x^{(j)} \in \xcal, \\
	G_{i,j} &= \left\{x^{*T} A x^* - x^{(i)T} A x^{(j)} < \epsilon/2\right\}, \quad &\forall i,j \text{ with } i \neq j, \text{ for which } x^{(i)},x^{(j)} \in \xcal.
	\end{align*}
	Observe that $E_{i} \cap F_{i,j} \subseteq G_{i,j}$. Now, let each of $k$ i.i.d. random variables be drawn from $x^*$. The sample space for each is $[p]$. For any $x^{(i)},x^{(j)} \in \Delta_p(k)$, the expectation of $x^{(i)T} A x^*$ is $x^{*T} A x^*$, and the expectation of $x^{(i)T} A x^{(j)}$ (for fixed $x^{(i)}$) is $x^{(i)T} A x^*$. Let us denote $r := |\xcal|  = 2/\epsilon$. By using a H{\"o}ffding bound \cite{H63}, we get 
	\begin{align*}
	\text{Pr}\{\overline{E_{i}}\} &\leq e^{-k \epsilon^2 /8}, \quad \forall i \text{ for which } x^{(i)} \in \xcal, \text{ and}\\
	\text{Pr}\{\overline{F_{i,j}}\} &\leq e^{-k \epsilon^2 /8}, \quad \forall i,j \text{ with } i \neq j, \text{ for which } x^{(i)},x^{(j)} \in \xcal.
	\end{align*}
	%	Therefore
	%	\begin{align*}
	%	\text{Pr}\{\overline{G_{i,j}}\} &\leq 2e^{-k \epsilon^2 /8}, \quad \forall i,j \in R \quad \text{with} \quad i \neq j.
	%	\end{align*}
	
	Consider now the event $H$ that captures the condition that needs to be satisfied by the lemma. It is
	\begin{align*}
	H &= \bigcap_{\substack{i,j \in \xcal \\ i \neq j}}G_{i,j}.
	\end{align*}
	Therefore
	\begin{align*}
	\overline{H} = \bigcup_{\substack{i,j \in \xcal \\ i \neq j}}\overline{G_{i,j}} 
	\subseteq \bigcup_{\substack{i \in \xcal }}\overline{E_{i}} \bigcup_{\substack{i,j \in \xcal \\ i \neq j}}\overline{F_{i,j}}.
	\end{align*}
	Hence
	\begin{align*}
	\text{Pr}\{\overline{H}\} &\leq r e^{-k \epsilon^2 /8} + r(r-1) e^{-k \epsilon^2 /8} \\
	&=r^2 e^{-k \epsilon^2 /8} \\
	&< 1.
	\end{align*}
\end{proof}

The above strict inequality means that $\text{Pr}\{H\}>0$, therefore, there exists a set $\xcal$ that satisfies the statement of the lemma.

The following theorem corresponds to the general Lemma \ref{lem:k-TSD}, for the case $\alpha = \gamma = 1$, $d=2$.

\begin{theorem}\label{thm: x}
	There exists a $\frac{32 \ln (3/\epsilon)}{\epsilon^3}$-uniform vector $x$, such that $v^* - x^T A x < \epsilon $ .
\end{theorem}

\begin{proof}
	Consider the multiset $\xcal$ of $\Delta_p(k)$ of Lemma \ref{lem: R}, and recall that $r := |\xcal| = 2/\eps$. Let us create the vector
	\begin{align*}
	x := \frac{1}{r} \sum_{i \in \xcal}x^{(i)} .
	\end{align*} 
	Then, it is 
	\begin{align*}
	x^{*T} A x^* - x^{T} A x &= 	x^{*T} A x^* - \left(\frac{1}{r} \sum_{x^{(i)} \in \xcal}x^{(i)T}\right) A \left(\frac{1}{r} \sum_{x^{(i)} \in \xcal}x^{(i)}\right) \\
	&= x^{*T} A x^* - \frac{1}{r^2} \sum_{x^{(i)},x^{(j
			)} \in \xcal} x^{(i)T} A x^{(j)} \\
	&= x^{*T} A x^* - \frac{1}{r^2} \left( \sum_{\substack{x^{(i)},x^{(j)} \in \xcal \\ i \neq j}} x^{(i)T} A x^{(j)} + \sum_{x^{(i)} \in \xcal} x^{(i)T} A x^{(i)} \right) \\
	&= \frac{1}{r^2} \left( r(r-1) x^{*T} A x^* - \sum_{\substack{x^{(i)},x^{(j)} \in \xcal \\ i \neq j}} x^{(i)T} A x^{(j)} + r x^{*T} A x^* - \sum_{x^{(i)} \in \xcal} x^{(i)T} A x^{(i)} \right) \\
	&< \frac{1}{r^2} \left( r(r-1) \frac{\epsilon}{2} + r  \right) \\
	&\leq \frac{\epsilon}{2} + \frac{1}{r} \\
	&= \epsilon,
	\end{align*}
	where the second to last inequality is implied from Lemma \ref{lem: R} which applies for every $x^{(i)},x^{(j)} \in \xcal$ when $i \neq j$, and from the fact that $x^{*T} A x^* -  x^{(i)T} A x^{(i)}$ is upper bounded by 1 for every $x^{(i)} \in \xcal$ (recall that the entries of $A$ are in $[0,1]$).
	
	The proof is concluded by observing that the vector $x$ we created is a $kr$-uniform vector, for $k = 16 \ln(3/\epsilon)/\epsilon^2$ and $r = 2/\epsilon$.
\end{proof}

\begin{corollary}\label{cor:PTAS_SQP}
	There is a PTAS for \textsc{sqp}.
\end{corollary}

\begin{proof}
	By Theorem \ref{thm: x}, since the desired probability vector $x$ that is suitable for the approximation is the mean of $r$ many $k$-uniform vectors, $x$ is $kr$-uniform. Therefore, it can be found by exhaustively searching through all possible multisets of $[p]$ created by sampling with replacement $kr = 32 \ln (3/\epsilon) / \epsilon^3$ times. The number of all those possible multisets is $\binom{p+kr-1}{kr} \in O(p^{kr})$. For each multiset, i.e. vector $x$ that the search algorithm takes into account, it picks the one that makes $x^T A x$ maximum. This value is guaranteed to be $\epsilon$-close to $v^*$ by Theorem \ref{thm: x}.
	
	Hence, if we desire a $(1-\eps)$-approximation of \textsc{sqp} \textit{in the weak sense} according to Definition 2.2 of \cite{dK08}, the described algorithm runs in time $O\left(p^{ \ln (\frac{3}{\eps}) / \eps^3}\right)$.
\end{proof}

\subsection{The general proof}

\subsubsection{Problems with multilinear constraints}

We begin by considering constrained \eetr problems where the Boolean
formula $F$ consists of tensor-multilinear polynomial constraints. 
We will use $\tml(A,x_1,\ldots,x_n)$ to denote a tensor-multilinear polynomial with 
$n$ variables and coefficients defined by tensor $A$ of size $\times_{j=1}^{n} p$. 
Formally,
\begin{align*}
\tml(A, x_1,\ldots,x_n) = \sum_{i_1 \in [p]}\cdots\sum_{i_n \in [p]} x_1(i_1)\cdot\ldots\cdot x_n(i_n)\cdot a(i_1, \ldots, i_n) + c.
\end{align*}
We will use $\alpha$ to denote the maximum entry of tensor $A$ in the  absolute value sense
and $\gamma$ to denote the infinite norm of the convex set that constrains the
variables.

%
%
%
%\subsection{Multilinear constraints}
%Firstly, we study convex-\etr problems with multilinear constraints, i.e. 
%every polynomial constraint in $\ccal_m^n$ is defined via a multilinear polynomial.
%%The idea is to replace every multilinear constraint with a set of linear constraints that
%%depend only on one variable so then the combination of  Hoeffding's inequality and 
%%the union bound will 
%
%%We will use $ML:=\langle \Pi(l,m,p), \ycal\rangle$ to denote the \etr problem 
%%we study in this section (ML for multilinear). 
%So, let $\Pi_{ML} = \langle \ccal_m^n, \ycal_m^n \rangle$ be a convex-\etr 
%problem, let $\gamma = \max_i \gamma_{Y_i}$, and let $\ahat$ be the maximum 
%absolute value of the coefficients of the multilinear constraints.

\begin{lemma}
	\label{lem:k-TML}
	Let $F$ be a Boolean formula with $n$ variables and $m$ tensor-multilinear polynomial
	constraints and let \ycal be a convex set in the variables space. If the constrained \etr problem 
	defined by $\exact(F)$ and \ycal has a solution, then the constrained \eetr problem defined by $F$ and \ycal has a $k$ uniform solution where 
	\begin{align*}
	k = \frac{2\cdot \alpha^2 \cdot \gamma^2 \cdot n^2\cdot \ln (3 \cdot n \cdot m)}{\eps^2}.
	\end{align*}
\end{lemma}

\begin{proof}
	Let $(x^*_1, x^*_2, \ldots, x^*_n) \in \ycal$ be a solution for $\exact(F)$. 
	Since we assume the \ycal is the convex hull of $c_1, \ldots, c_l$ any $x \in \ycal$ can be
	written as a convex combination of the $c_i$'s, i.e., $x = \sum_{i \in [l]} a_i \cdot c_i$, 
	where $a_i \geq 0$ for every $i \in [l]$, and $\sum_{i\in [l]} a_i = 1$. 
	Observe, $a = (a_1, \ldots, a_{l})$ corresponds to a probability distribution over 
	$c_1, \ldots, c_l$, where vector $c_i$ is drawn with probability $a_i$, and $x$ can be 
	thought of as the mean of $a$. So, we can ``sample'' a point by sampling over $c_i$'s
	according to the probability that defines this point.

	For every $i \in [n]$, let $x'_i$ be a $k$-uniform vector sampled independently from $x^*_i$.
	To prove the lemma, we will show that, because of the choice of $k$, with positive 
	probability the sampled vectors satisfy every constraint of the \eetr problem. 
	Then, by the probabilistic method the lemma will follow.
	
	Let $\tml_j(A_j,x_1, \ldots,x_n)$ be a multilinear polynomial that defines a 
	constraint of $F$. For every $j \in [m]$ we define the 
	following event
	\begin{align}
	\label{eq:ml-c}
	|\tml_j(A_j,x'_1, \ldots, x'_n) -\tml_j(A_j,x^*_1, \ldots,x^*_n)| \leq \eps.
	\end{align}
	Observe that if $x'_1, \ldots, x'_n$ satisfy inequality \eqref{eq:ml-c} for every $j \in [m]$, 
	then the lemma follows.
	
	For every $j \in [m]$, we replace the corresponding event \eqref{eq:ml-c} with 
	$n$ events that are {\em linear} in each variable. For notation simplicity, let us denote by $ML_j^i$ the 
	multilinear polynomial $TML_j(A_j,x_1, \ldots, x_n)$ in which we have additionally set 
	$x_1=x'_1, x_2=x'_2, \ldots, x_i=x'_i$ and $x_{i+1}=x^*_{i+1}, x_{i+2}=x^*_{i+2}, \ldots, 
	x_n=x^*_n$. Furthermore, let $ML_j^0 =  ML_j(A_j,x^*_1, \ldots,x^*_n)$.
	Then, for every $i \in [n]$ consider the event
	\begin{align}
	\label{eq:l-c}
	|ML_j^i - ML_j^{i-1}| \leq \frac{\eps}{n}.
	\end{align}
	Observe that, if for a given $j \in [m]$ all $n$ events defined in \eqref{eq:l-c} are satisfied, then by the 
	triangle inequality, the corresponding event \eqref{eq:ml-c} is satisfied as well.
	
	Consider now $ML_j^i$. This can be seen as a random variable that depends on 
	the choice of $x'_i$ and takes values in $[-\gamma\cdot \alpha, \gamma\cdot \alpha]$. 
	But recall that the $x'_i$'s are sampled from $x^*_i$ using $k$ samples, and that they are
	mutually independent, so $\ebb\big[ ML_j^i \big] = ML_j^{i-1}$. 
	Thus, we can bound the probability that a constraint~\eqref{eq:l-c} is not 
	satisfied, i.e. bound the probability that $|ML_j^i - ML_j^{i-1}| > \frac{\eps}{n}$, using
	Hoeffding's inequality \cite{H63}. So,
	\begin{align}
	\nonumber
	\pr\left(\left|ML_j^i - ML_j^{i-1}\right| > \frac{\eps}{n}\right) & = 
	\pr\left(\left|ML_j^i - \ebb\big[ ML_j^i \big] \right| > \frac{\eps}{n}\right)\\
	\nonumber
	& \leq 2\cdot exp\left(- \frac{2\cdot k^2 \cdot \left(\frac{\eps}{n}\right)^2}{4 \cdot k \cdot \gamma^2 \cdot \alpha^2} \right)\\
	\label{eq:hoef-bound-l}
	& = 2\cdot exp\left(- \frac{k \cdot \eps^2}{2 \cdot n^2 \cdot \gamma^2 \cdot \alpha^2} \right).
	\end{align}
	Recall, that we have $n \cdot m$ events of the form \eqref{eq:l-c}.
	We can bound the probability that any of those events is violated, via the union bound. 
	So, using~\eqref{eq:hoef-bound-l} and the union bound, the probability that 
	any of these events is violated is upper bounded by 
	\begin{align}
	\label{eq:union-bound-l}
	2 \cdot m\cdot n \cdot exp\left(- \frac{k \cdot \eps^2}{2 \cdot n^2 \cdot \gamma^2 \cdot \alpha^2} \right).
	\end{align}
	Hence, if the value of \eqref{eq:union-bound-l} is strictly less than 1, then there are 
	$x'_1, \ldots, x'_m$ such that all of the $n \cdot m$ events of \eqref{eq:l-c} are realized with positive probability, therefore the events of \eqref{eq:ml-c} are realized with positive probability
	and thus the lemma follows. 
	By requiring \eqref{eq:union-bound-l} to be strictly less than 1, and solving for $k$ we get
	$$k > \frac{ 2 \cdot \alpha^2 \cdot \gamma^2 \cdot n^2 \cdot \ln (2\cdot n \cdot m)}{\eps^2}$$
	which holds, by our choice of $k$.

\end{proof}

\subsubsection{Problems with a standard degree $d$ constraint}

We now consider constrained \eetr problems with 
\emph{exactly one} tensor polynomial constraint of standard degree $d$. 
We will use $\tsd(A,x,d)$ to denote a standard degree $d$ tensor-polynomial with 
coefficients defined by the $\times_{j=1}^{d} p$ tensor $A$. 
Here, $d$ identical vectors $x$  are applied on A. Formally,
\begin{align*}
\tsd(A,x,d) = \sum_{i_1 \in [p]}\cdots\sum_{i_d \in [p]} x(i_1)\cdot\ldots\cdot x(i_d)\cdot a(i_1, \ldots, i_d) + c.
\end{align*}

%Recall that, $\alpha$ denotes the maximum entry (in absolute value) of tensor $A$ and 
%that $\gamma$ is the infinity norm of the domain \ycal we consider. 
To prove the following lemma we consider the variable $x$ to be defined as the average of 
$r = O(\frac{\alpha^2 \cdot \gamma^d \cdot d^2}{\eps})$ variables. This allows us
to ``break'' the standard degree $d$ tensor polynomial to a sum of multilinear 
tensor polynomials and to a sum of not-too-many multivariate polynomials. 
Then, the choice of $r$ allows us to upper bound the error occurred by the multivariate
polynomials by $\frac{\eps}{2}$. Then, we observe that in order to prove the lemma we 
can write the sum of multilinear tensor polynomials as an $\frac{\eps}{2}$-\etr problem
with $r$ variables and roughly $r^d$ multilinear constraints. This allows us to use 
Lemma~\ref{lem:k-TML} to complete the proof.

%%%%%%%%%%%%%%%%%%%%%%%%%%%%%%%%%

%In particular, to prove Lemma \ref{lem:k-TSD} we will first prove the following auxiliary lemma.
%
\begin{lemma}
	\label{lem:standard-d-new}
	Let $F$ be a Boolean formula with one variable and one tensor-polynomial constraint
	of standard degree $d$, let \ycal be a bounded convex set, and let $r=\frac{2 \cdot \alpha^2 \cdot \gamma^d \cdot d^2 }{\eps}$.
	If the constrained \etr problem $\exact(F)$ has a solution in \ycal, then there 
	exists a satisfiable constrained $\frac{\eps}{2}$-\etr problem $\Pi_{ML}$ with $r$
	variables, where each variable is a $k$-uniform vector for $k=\frac{16 \cdot \alpha^4 \cdot \gamma^d \cdot d^4}{\eps^3}$. The Boolean formula of $\Pi_{ML}$ is the conjunction of $\prod_{i=0}^{d-1}(r-i)$ tensor multilinear constraints, and every 
	solution of $\Pi_{ML}$ in \ycal can be transformed to a solution for the constrained
	\eetr problem defined by $F$ and \ycal.
\end{lemma}

\begin{proof}
	Assume that $x^* \in \ycal$ is a solution for $F$. Let $\tsd(A,x,d)$ denote the tensor
	polynomial of standard degree $d$ used in $F$. 
	For notation simplicity, let $\tsd(A,x,d) = A(x^d)$.
	Create $r$ new $k$-uniform variables $x_1, \ldots, x_r \in \ycal(k)$ by sampling each one from $x^*$, where $\ycal(k)$ is the discretized set made from $\ycal$ by using $k$-uniform vectors, and set
	$x = \frac{1}{r}(x_1+ \ldots + x_r)$. 
	%	Each $x_i$, $i \in [r]$ is a $k$-uniform vector, where $k$ is as in the statement of Lemma \ref{lem:k-TML} for $n=m=1$ as the current lemma's statement indicates.
	Let $\xcal= \bigcup_{i=1}^r \{x_i\}$ be a \textit{multiset} of $\ycal(k)$ with cardinality $r$, meaning that multiple copies of an element of $\ycal(k)$ are allowed in $\xcal$. In the sequel we will treat the elements of $\xcal$ as distinct, even though some might correspond to the same element of $\ycal(k)$.
	Then, note that $A(x^d)$ can be written as a sum of simple tensor-multivariate 
	polynomials where some of them are multilinear and have as variables  $x_1, \ldots, x_r$. 
	Now, let $\scal$ be the set of all ordered $d$-tuples that can be made by drawing $d$ elements from $\xcal$ with replacement. Formally, 
	$\scal = \{(\xhat_{1}, \ldots, \xhat_{d}): \xhat_1, \ldots, \xhat_d \in \xcal\}$.
	Let us also define $\scal_d$ to be the set of all ordered $d$-tuples that can be made by drawing $d$ elements from $\xcal$ without replacement. Formally, $\scal_d = \{(\xhat_{1}, \ldots, \xhat_{d}): \xhat_{1}, \ldots, \xhat_{d} \in \xcal, \quad \xhat_{1},\ldots, \xhat_{d}~\text{are pairwise different}\}$, and observe that $|S_d| =  \prod_{i=0}^{d-1}(r-i)$. 
	So, any element of $\scal_d$, combined with tensor $A$, produces a multilinear polynomial.
	Hence, using the notation introduced, we get that $|A(x^d) - A(x^{*^d})|$ is less than or equal to the sum of the following two sums
	\begin{align}
	\label{eq:sum1}
	& \frac{1}{r^d} \sum_{(\xhat_{1}, \ldots, \xhat_{d}) \in \scal_d} \left| A(\xhat_{1}, \ldots, \xhat_{d})  - A(x^{*^d}) \right| \qquad \text{and}\\
	\label{eq:sum2}
	& \frac{1}{r^d} \sum_{(\xhat_{1}, \ldots, \xhat_{d}) \in \scal - \scal_d} \left| A(\xhat_{1}, \ldots, \xhat_{d})  -  A(x^{*^d}) \right|.
	\end{align}
	Observe, $|\scal-\scal_d| = r^d - |\scal_d|$ and that $\left| A(\xhat_{1}, \ldots, \xhat_{d})  -  A(x^{*^d}) \right| \leq \gamma^d \cdot \alpha$ for every 
	$A(\xhat_{1}, \ldots, \xhat_{d})$. Then, for the sum given in~\eqref{eq:sum2} 
	we get
	\begin{align*}
	\frac{1}{r^d} \sum_{(\xhat_{1}, \ldots, \xhat_{d}) \in \scal - \scal_d} & \left| A(\xhat_{1}, \ldots, \xhat_{d})  -  A(x^{*^d}) \right| \\
	& \leq \left(1 - \frac{r\cdot (r-1) \cdots (r-d+1)}{r^d} \right)\cdot \gamma^d \cdot \alpha \\
	& \leq \left(1 - \left(1 -\frac{1}{r}\right) \left(1 -\frac{2}{r}\right) \cdot \left(1 -\frac{d-1}{r}\right) \right)\cdot \gamma^d \cdot \alpha\\
	& \leq \left(1 - \left(1 -\frac{d-1}{r}\right)^{d-1} \right)\cdot \gamma^d \cdot \alpha \\
	& \leq \left(1 - \left(1 -\frac{(d-1)^2}{r}\right)\right)\cdot \gamma^d \cdot \alpha \quad \text{(Bernoulli's inequality)} \\ 
	& = \frac{(d-1)^2}{r} \cdot \gamma^d \cdot \alpha\\
	& \leq \frac{\eps}{2}.
	\end{align*}
	Hence, in order for the original constraint to be satisfied, it suffices to satisfy the constraint 
	\begin{align}
	\label{eq:xd-help-12}
	\frac{1}{r^d} \sum_{(\xhat_{1}, \ldots, \xhat_{d}) \in \scal_d} \left| A(\xhat_{1}, \ldots, \xhat_{d})  - A(x^{*^d}) \right| \leq \frac{\eps}{2}.
	\end{align}
	Observe that $|\scal_d| =  \prod_{i=0}^{d-1}(r-i) < r^d$, therefore, instead of the constraint~\eqref{eq:xd-help-12}, 
	it suffices to satisfy the following $|S_d|$ constraints (we introduce one 
	constraint for every $(\xhat_{1}, \ldots, \xhat_{d}) \in \scal_d$)
	\begin{align}
	\label{eq:xd-help-13}
	\left| A(\xhat_{1}, \ldots, \xhat_{d})  - A(x^{*^d}) \right| \leq \frac{\eps}{2}.
	\end{align}
	Note that each constraint \eqref{eq:xd-help-13} is the relaxed by $\eps/2$ version of a constraint with a multilinear function equal to 0; multilinearity is due to the fact that $\xhat_{1}, \ldots, \xhat_{d}$ are pairwise different by definition of the set $S_d$. The proof is completed by using Lemma \ref{lem:k-TML} for $n=d$, $m = |S_d|$ and $\eps / 2$ instead of $\eps$ to show that indeed there exists a collection $S_d$ of tuples $\xhat_{1}, \ldots, \xhat_{d}$, where each $\xhat_{i}$, $i \in [d]$ is a $k$-uniform vector with $k \geq \frac{8 \cdot \alpha^2 \cdot \gamma^2 \cdot d^2 (d+2) \cdot \ln r}{\eps^2}$
	such that all $|S_d|$ constraints of \eqref{eq:xd-help-13} are satisfied. The latter inequality is true by our choice of $k$ and $r$.

\end{proof}

Now we can prove the following lemma.

\begin{lemma}
	\label{lem:k-TSD}
	Let $F$ be a Boolean formula with variable $x$ and one tensor-polynomial constraint
	of standard degree $d$,  and let \ycal be a bounded convex set. If the 
	constrained \etr problem defined by $\exact(F)$ and \ycal has a solution, then the 
	constrained \eetr problem defined by $F$ and \ycal has a $k$-uniform solution where 
	\begin{align*}
	k = \frac{32 \cdot \alpha^6 \cdot \gamma^{2d} \cdot d^6}{\eps^4}.
	\end{align*}
\end{lemma}

\begin{proof}
	First, we use Lemma~\ref{lem:standard-d-new} to construct the constrained $\frac{\eps}{2}$-\etr problem $\Pi_{ML}$ with tensor-multilinear constraints.
	Recall that $\Pi_{ML}$ has $r=\frac{2\cdot \alpha^2 \cdot \gamma^d \cdot d^2}{\eps}$
	variables and if $\Pi_{ML}$ is
	satisfiable, then there exist $\frac{k}{r}$-uniform vectors 
	$\xhat_1 \in \ycal, \ldots, \xhat_r \in \ycal$ that $\eps/2$-satisfy $\Pi_{ML}$.  
	Then, let us construct the $k$-uniform vector $\xhat = \frac{1}{r} \cdot (\xhat_1+ \ldots+ \xhat_r)$. Note that, according to Lemma~\ref{lem:standard-d-new}, it is
	\begin{align*}
	|A({\xhat}^d) - A(x^{*^d})| &\leq \frac{1}{r^d} \sum_{(\xhat_{1}, \ldots, \xhat_{d}) \in \scal_d} \left| A(\xhat_{1}, \ldots, \xhat_{d})  - A(x^{*^d}) \right| \\
	&\qquad + \frac{1}{r^d} \sum_{(\xhat_{1}, \ldots, \xhat_{d}) \in \scal - \scal_d} \left| A(\xhat_{1}, \ldots, \xhat_{d})  -  A(x^{*^d}) \right| \\
	&\leq \frac{\eps}{2} + \frac{\eps}{2} \\
	&=\eps.
	\end{align*}
	This completes the proof of the lemma.
\end{proof}

\subsubsection{Problems with simple multivariate constraints}
%Next we show how to handle problems in convex-\etr with simple
%tensor-multivariate polynomial constraints. The idea is to reduce every 
%simple multivariate constraint to a set of tensor standard degree $d$
%constraints, then use Lemma~\ref{lem:standard-d-new} to produce
%a set of multilinear constraints, and finally use Lemma~\ref{lem:multi-k-uniform}
%to show the existence of $k$-uniform vectors that \eps-satisfy 
%these multilinear constraints.
%
We now assume that we are given a constraint-\eetr
problem defined by a Boolean formula $F$ of tensor simple multilinear polynomial 
constraints and a bounded convex set \ycal. As before $\gamma = \|\ycal\|_\infty$ and let
$\alpha$ be the maximum absolute value of the coefficients of the constraints. 
We will say that the constraints are of maximum degree $d$ if $d$ is the maximum degree among all variables.  
%$\Pi_{\stm} = \langle \ccal_m^n, \ycal_m^n \rangle$, where $\ycal_m^n$ is a collection of convex hulls $Y_i$ of respective finite sets $Z_i \subset \reals^n$ for $i \in [m]$,
%$\gamma = \max_i \gamma_{Y_i}$, and let $\ahat$ be the maximum 
%absolute value of the coefficients of the simple multivariate 
%constraints. Furthermore, for every $i \in [m]$ denote by $d^*_i$ 
%the maximum degree of variable $x_i$ over all the simple 
%multivariate constraints of $\Pi_{\stm}$, and let $d^*=\max_i d^*_i$.
%Similarly to~\eqref{eq:sd-params} we define 
%%\begin{align*}
%%r_{(SM,\eps)} := \frac{2\cdot m \cdot (d^*-1)^2\cdot \gamma^{d^*} \cdot \ahat^2}{\eps}
%%\qquad  \text{and} \qquad |\scal_{d^*}| = \prod_{i=0}^{d^*-1}(r-i).
%%\end{align*}
%%
%\begin{align*}
%r_{(SM,\eps)} := \frac{2\cdot m \cdot (d^*-1)^2\cdot \gamma^{d^*} \cdot \ahat^2}{\eps}
%\qquad  \text{and} \qquad \zeta^* = \prod_{i=0}^{d^*-1}(r-i).
%\end{align*}
%Observe, $r_{(SM,\eps)}$ differs from $r$ of~\eqref{eq:sd-params} not only
%in $d^*$, but also in the multiplication with $m$.
%
The main idea of the proof of the following lemma is to rewrite the problem as an equivalent problem with standard degree
$d$ constraints and then apply Lemmas~\ref{lem:k-TSD} and~\ref{lem:k-TML} to 
derive the bound for $k$.

\begin{lemma}
	\label{lem:k-TSM}
	Let $F$ be a Boolean formula with $n$ variables and $m$ simple tensor-multivariate
	polynomial constraints of maximum degree $d$ and let \ycal be a bounded convex set in the 
	variables space. If the constrained \etr problem defined by $\exact(F)$ and \ycal has a
	solution, then the constrained \eetr problem defined by $F$ and \ycal has a $k$-uniform
	solution where
	\begin{align*}
	%	k = \frac{48 \cdot \alpha^6 \cdot \gamma^{2d+2} \cdot d^5 \cdot n^6 \cdot \ln(2 \cdot 
	%		\alpha \cdot \gamma \cdot d \cdot n \cdot m)}{\eps^4}.
	k = \frac{512 \cdot \alpha^6 \cdot \gamma^{2d+2} \cdot n^6 \cdot d^6 \cdot \ln ( 2 \cdot \alpha' \cdot \gamma' \cdot d \cdot n \cdot m )}{\eps^5},
	\end{align*}
	where $\alpha':=\max(\alpha, 1), \gamma':=\max(\gamma, 1)$.
\end{lemma}

\begin{proof}
	Let $x^*_1, \ldots, x^*_n$ be a solution for $\exact(F)$ 
	%	Let 
	%	$r = \frac{2 \cdot \alpha^2  \cdot \gamma^{d} \cdot d^2}{\eps}$ 
	%	and let $\xhat^i_1,\ldots,\xhat^i_r$ be $\frac{k}{r}$-uniform vectors sampled 
	%	from $x^*_i$.
	%	Define $x'_i = \frac{1}{r}(\xhat^i_1+ \ldots + \xhat^i_r)$.
	and let $x'_i$, $i \in [n]$ be a variable $k$-uniform vector sampled from $x_{i}^*$. We will prove that if $k$ equals at least the quantity of the statement of the lemma, then there exist vectors $x'_{1}, \dots, x'_{n}$ that constitute a solution to the constrained \eetr problem defined by $F$ and $\ycal$.
	
	Consider the $j$-th constraint where $j \in [m]$ defined by the simple tensor-multivariate 
	polynomial $\stm(A_j,x_1^{d_{j1}}, \ldots,x_n^{d_{jn}})$. We will use the
	same technique we used in Lemma~\ref{lem:k-TML} to create $n$ constraints,
	where constraint $i \in [n]$ is defined via a simple degree $d_{ji}$ polynomial.
	Again, for notation simplicity for every $i \in [m]$ we use $\stm_j^i$ to denote
	the polynomial $\stm(A_j,x_1^{d_{j1}}, \ldots,x_n^{d_{jn}})$ where we set 
	$x_1 = x'_1, \ldots, x_i = x'_i$ and $x_{i+1} = x^*_{i+1}, \ldots, x_n = x^*_n$.
	Let $\stm_j^0 := \stm(A_j,(x^*_1)^{d_{j1}}, \ldots,(x^*_n)^{d_{jn}})$.
	Then, for every $j \in [m]$ we define the following $n$ constraints
	\begin{align}
	\label{eq:sm-to-sd}
	|\stm_j^i - \stm_j^{i-1}| \leq \frac{\eps}{n}.
	\end{align}
	
	Observe that for some $j \in [m]$, every constraint $i$ of the form~\eqref{eq:sm-to-sd} defines a simple degree $d_{ji}$ polynomial with respect to variable $x'_i$. Furthermore, observe 
	that if every such constraint is satisfied, then the initial constraint defined
	by $\stm(A_j,x_1^{d_{j1}}, \ldots,x_n^{d_{jn}})$ is satisfied too. Then, we convert each 
	such constraint to a set of $\prod_{i=0}^{d-1}(r-i)$ multilinear constraints with 
	$r = \frac{2 \cdot \alpha^2  \cdot \gamma^{d} \cdot d^2}{\eps}$ variables, using Lemma \ref{lem:standard-d-new} where we demand that every 
	multilinear constraint is $\frac{\eps}{2n}$-satisfied (we restrict the current $\frac{\eps}{n}$ to half of it in order to use Lemma \ref{lem:standard-d-new}). The proof is then 
	completed by using Lemma~\ref{lem:k-TML} where we observe that we have 
	$r \cdot n = \frac{2 \cdot \alpha^2  \cdot \gamma^{d} \cdot d^2 \cdot n}{\eps}$ 
	variables and $\prod_{i=0}^{d-1} (r-i)\cdot n \cdot m < r^d \cdot n \cdot m $ constraints 
	and we set $\eps$ to $\frac{\eps}{2n}$.
	
	To arrive to the actual size $k$ of the required uniform vector, we start from the size $k'$ prescribed by Lemma \ref{lem:k-TML} and sequentially set proper values for the parameters as dictated by our method for transforming the constraints. We have
	
	\begin{align*}
	k' &= \frac{2\cdot \alpha^2 \cdot \gamma^2 \cdot n^2\cdot \ln (3 \cdot n \cdot m)}{\eps^2} \\
	&= 	\frac{8 \cdot \alpha^6 \cdot \gamma^{2d+2} \cdot n^2 \cdot d^4 \cdot \ln (6 \cdot \alpha^2 \cdot \gamma^d \cdot d^2 \cdot n \cdot m / \eps)}{\eps^4} \qquad (n \leftarrow n \cdot r) \\
	&= 	\frac{8 \cdot \alpha^6 \cdot \gamma^{2d+2} \cdot n^2 \cdot d^4 \cdot \ln (6 \cdot \alpha^{2d+2} \cdot \gamma^{d^2 + d} \cdot d^{2d + 2} \cdot n^2 \cdot m / \eps^{d+1})}{\eps^4} \qquad (m \leftarrow r^d \cdot n \cdot m) \\
	&= 	\frac{128 \cdot \alpha^6 \cdot \gamma^{2d+2} \cdot n^6 \cdot d^4 \cdot \ln (6 \cdot 2^{d+1} \cdot \alpha^{2d+2} \cdot \gamma^{d^2 + d} \cdot d^{2d + 2} \cdot n^{d+3} \cdot m / \eps^{d+1})}{\eps^4} \qquad (\eps \leftarrow \frac{\eps}{2n}) \\
	&\leq 	\frac{128 \cdot \alpha^6 \cdot \gamma^{2d+2} \cdot n^6 \cdot d^4 \cdot \ln ( 2 \cdot \max(\alpha, 1) \cdot \max(\gamma, 1) \cdot d \cdot n \cdot m / \eps)^{4d^2}}{\eps^4} \qquad (\text{for any } d \geq 1)  \\
	&= 	\frac{512 \cdot \alpha^6 \cdot \gamma^{2d+2} \cdot n^6 \cdot d^6 \cdot \ln ( 2 \cdot \max(\alpha, 1) \cdot \max(\gamma, 1) \cdot d \cdot n \cdot m / \eps)}{\eps^4}  \\
	&\leq 	\frac{512 \cdot \alpha^6 \cdot \gamma^{2d+2} \cdot n^6 \cdot d^6 \cdot \ln ( 2 \cdot \max(\alpha, 1) \cdot \max(\gamma, 1) \cdot d \cdot n \cdot m )}{\eps^5}.
	\end{align*}	
	We want $k \geq k'$, therefore it suffices to bound from below $k$ by the upper bound of $k'$. This completes the proof.
\end{proof}

\subsubsection{Putting everything together}

%In this section we show our most general result; we show how to 
%deal with problems in convex-\etr with multivariate constraints.
%Given a convex-\etr problem $\Pi_{\tmv}$ with multivariate constraints, 
%we reduce it to a problem $\Pi_{\stm}$ with simple multivariate constraints,
%so then we can use Lemma~\ref{lem:sm-k-uni} to derive a value for $k$ 
%such that the existence of a $k$-uniform approximate solution is guaranteed.
%
%\begin{theorem}[Theorem~\ref{thm:main}]
%%\label{thm:main}
%Let $F$ be an $\eps$-\etr instance with $n$ variables and $m$ multivariate-polynomial
%constraints of maximum length $t$ and maximum degree $d$, 
%constrained by a convex set defined by $c_1, c_2, \ldots, c_l \in \reals^n$. 
%Let $\alpha$ be the maximum coefficient of any term in $F$, and let $\gamma =
%\max\|c_i\|_\infty$. If $\exact(F)$ has a solution in $\conv(c_1, c_2, \ldots,
%c_l)$, then $F$ has a $k$-uniform solution in $\conv(c_1, c_2, \ldots, c_l)$
%where
%\begin{align}
%\label{eq:k}
%k = \frac{64 \cdot \gamma^{4d+2} \cdot \alpha^{10} \cdot d^9 \cdot n^4 \cdot m^2 \cdot t^8 \cdot \ln(2 \cdot d \cdot \gamma \cdot \alpha \cdot n \cdot m \cdot t)}{\eps^6}.
%\end{align}
%\end{theorem}
\begin{proof}
	For the final step of the proof of Theorem \ref{thm:main}, assume that $x^*_1, \ldots, x^*_n \in \ycal$ is a solution for $\exact(F)$. Consider now 
	a multivariate constraint $i \in [m]$ of $F$ defined by $TMV_i(x_1, \ldots, x_n)$. Firstly, 
	we replace this constraint by 
	\begin{align}
	\label{eq:tmv-help}
	|TMV_i(x_1, \ldots, x_n) - TMV_i(x^*_1, \ldots, x^*_n)| \leq \eps.
	\end{align}
	Then, replace constraint \eqref{eq:tmv-help} by $t$ constraints of the form
	\begin{align}
	\label{eq:stm-help}
	|STM_{i,j}(x_1, \ldots, x_n) - STM_{i,j}(x^*_1, \ldots, x^*_n)| \leq \frac{\eps}{t}
	\end{align}
	where $STM_{i,1}(x_1, \ldots, x_n), \ldots, STM_{i,t}(x_1, \ldots, x_n)$ are the simple
	tensor multivariate polynomials $TMV_i(x_1, \ldots, x_n)$ consists of. By the triangle 
	inequality we get that if all $t$ constraints given by~\eqref{eq:stm-help} hold, then 
	constraint \eqref{eq:tmv-help} holds as well. Hence, we can reduce the problem to an 
	equivalent problem with the same $n$ variables and  $m\cdot t$ constraints that all of them
	are simple tensor multivariate polynomials. So, we can apply Lemma~\ref{lem:k-TSM} where we 
	replace $m$ with $m\cdot t$ and $\eps$ with $\frac{\eps}{t}$. This completes the proof of the theorem.
\end{proof}

%%%%%%%%%%%%%%%%%%%%%%%%%%%%%%%%%%%%%%%%%%%%%%%%%%%%%%%%%%%%%%%
%%%%%%%%%%%%%%%%%%%%%%%%%%%%%%%%%%%%%%%%%%%%%%%%%%%%%%%%%%%%%%%
%%%%%%%%%%%%%%%%%%%%%%%%%%%%%%%%%%%%%%%%%%%%%%%%%%%%%%%%%%%%%%%
%
%\begin{theorem}
%\label{thm:main}
%Let $\Pi(l,m,n)$ be a satisfiable problem in Simple-convex \eetr. Then, for every 
%$\delta>0$ there exist $k$-uniform vectors $x'_1,\ldots, x'_m$ that $\delta$-satisfy
%every constraint of $\Pi(l,m,n)$, where $k=BLAH$.
%\end{theorem}
%%
%\begin{proof}
%Assume that $x^*_1, \ldots, x^*_m$ satisfy every constraint of $\Pi(l,m,n)$. 
%We will prove that there exist $k$-uniform vectors $x'_1,\ldots, x'_m$ that 
%$\delta$-satisfy every constraint of $\Pi(l,m,n)$. 
%
%Before we proceed with the formal proof let us briefly explain the route we
%will follow. Firstly, we will independently sample $m$ $k$-uniform vectors  
%$x'_1,\ldots, x'_m$ from 
% use the set of constraints of $\Pi(l,m,n)$, 
%vectors to 
%construct 
%
%To do this, we will prove that 
%for every constraint $A_i(x^{d_1}_1,\ldots, x^{d_m}_m)$, where $i \in [l]$ it holds that 
%$|A_i((x^*_1)^{d_1},\ldots, (x^*_m))^{d_m}) - A_i((x'_1)^{d_1},\ldots, (x'_m))^{d_m})| \leq \delta$.
%
%Observe,  
%
%
%\end{proof}

\section{Applications}
\label{sec:appl}

We now show how our theorems can be applied to derive new approximation
algorithms for a variety of problems. In order to conclude that Corollary \ref{cor:qptas} provides a PTAS or QPTAS for some given problem, one has to carefully determine the actual input size of the problem and show that the running time of the corollary's algorithm satisfies the PTAS or QPTAS definition.

\subsection{Constrained approximate Nash equilibria}

A {\em constrained} Nash equilibrium is a Nash equilibrium that satisfies some
extra constraints, like specific bounds on the payoffs of the players. 
Constrained Nash equilibria attracted the attention of many authors, who 
proved \NP-completeness for two-player games~\cite{GZ89,CS08,BM12} and 
\etr-completeness for three-player games~\cite{BM12,BM14,BM16,BM17,GMV+} for 
constrained \emph{exact} Nash equilibria. 

Constrained approximate equilibria have been studied, but so far only lower
bounds have been derived~\cite{ABC13,HK11,BKW15,DFS18,DFS17}. It has been
observed that sampling methods can give QPTASs for finding constrained
approximate Nash equilibria for certain constraints in two player
games~\cite{DFS18}.

By applying Theorem~\ref{thm:main}, we get the following result for games with number of players up to polylogarithmic in the number of pure strategies
(here n is the number of players): \emph{Any property of
	an approximate equilibrium that can be formulated in $\eps$-\etr where $\alpha$,
	$\gamma$, $d$, $t$ and $n$ are up to polylogarithmic in the number of pure strategies has a QPTAS.} This generalises past
results to a much broader class of constraints, and provides results for games
with more than two players, which had not previously been studied in this
setting. 

A game is defined by the set of players, the set of actions for every 
player, and the payoff function of every player. In normal form games, 
the payoff function is given by a multilinear function on a tensor of 
appropriate size. Consider an $n$-player game where every player has 
$l$-actions, and let $A_j$ denote the payoff tensor of player $j$ with elements
in $[0,1]$; $A_j$ has size $\times_{i=1}^n l$. 
The interpretation of the tensor $A_j$ is the following: the element 
$A_j(i_1, \ldots, i_n)$ of the tensor corresponds to the payoff of player
$j$ when Player 1 chooses action $i_1$, Player 1 chooses action $i_2$, and
so on. To play the game, every player $j$ chooses a probability distribution 
$x_j \in \Delta^l$, a.k.a. a {\em strategy}, over their actions. A collection 
of strategies is called {\em strategy profile}. The expected payoff of player 
$j$ under the strategy profile $(x_1, \ldots, x_n)$ is given by 
$ML(A_j, x_1, \ldots, x_n)$. For notation simplicity, let 
$u_j(x_j,x_{-j}) := ML(A_j, x_1, \ldots, x_n)$, where $x_{-j}$ is the strategy
profile of all players except player $j$.
A strategy profile $(x^*_1, \ldots, x^*_n)$ is a Nash equilibrium if for
every player $j$ it holds that $u_j(x^*_j,x^*_{-j}) \geq u_j(x_j,x^*_{-j})$
for every $x_j \in \Delta^l$, or equivalently $u_j(x^*_j,x^*_{-j}) \geq 
u_j(s_{p},x^*_{-j})$ for every possible $s_{p}$, where $s_{p}$ denotes the
case where player $j$ chooses their action $p$ with probability 1.

Our framework formally describes a broad family of constrained Nash equilibrium problems for which we can get a QPTAS.
\begin{theorem}
	\label{thm:eps-NE}
	Let $\Gamma$ be an $n$-player $l$-action normal form game $\Gamma$. Furthermore, let $F$ be a Boolean formula with $c \in \poly(l)$ TMV constraints 
	of degree $d$.
	If $n, d \in \polylog(l)$, then in quasi-polynomial time we can compute an approximate NE of $\Gamma$ 
	constrained by $F$, or decide that no such constrained approximate NE exists.
\end{theorem}
\begin{proof}
	Observe that we can write the problem of the existence of a constrained Nash 
	equilibrium as an \etr problem. The constraints of the problem will be the constraints 
	of $F$ plus the constraint
	$$u_j(s_{l},x_{-j}) - u_j(x_j,x_{-j}) \leq 0$$
	for every player $i \in [m]$ and every action $s_l$ of player $j$.
	
	Thus, we can use Theorem~\ref{thm:main} and complete the proof since we produced an
	\eetr problem with $m = c + n \cdot l = \poly(l)$ constraints, which is polynomial in the
	input size; $d$ and $t$ are polylogarithmic in $l$ by assumption (it always holds that $t \leq d$); $\gamma=1$ since every variable is 
	a probability distribution; $\alpha=1$ by the definition of normal form games.
\end{proof}

\subsection{Shapley games}

Shapley's stochastic games~\cite{S53} describe a two-player infinite-duration
zero-sum game. The game consists of $N$ \emph{states}. Each state specifies a
two-player $M \times M$ bimatrix game where the players compete over: (1) a reward
(which may be negative) that is paid by player two to player one, and (2) a
probability distribution over the next state of the game. So each round consists
of the players playing a bimatrix game at some state $s$, which generates a
reward, and the next state $s'$ of the game. The reward in round $i$ is
discounted by $\lambda^{i-1}$, where $0 < \lambda < 1$ is a \emph{discount
	factor}. The overall payoff to player 1 is the discounted sum of the infinite
sequence of rewards generated during the course of the game.

Shapley showed that these games are determined, meaning that there exists a
value vector $v$, where $v_s$ is the value of the game starting at state $s$. A
polynomial time algorithm has been devised for computing the value vector of a
Shapley game when the number of states $N$ is constant~\cite{HKLMT11}. However,
since the values may be irrational, this algorithm needs to deal with algebraic
numbers, and the \emph{degree} of the polynomial is $O(N)^{N^2}$, so if $N$ is even
mildly super-constant, then the algorithm is not polynomial.

Furthermore, Shapley showed that the value vector is the unique solution of a system of
polynomial optimality equations, which can be formulated in \etr. Any approximate
solution of these equations gives an approximation of the value vector, and
applying Theorem~\ref{thm:main} gives us a QPTAS. This algorithm works when $N
\in O(\sqrt[6]{\log M})$, which is a value of $N$ that prior work cannot handle.
The downside of our algorithm is that, since we require the solution to be
bounded by a convex hull defined by finitely many points, the algorithm only works when the value vector is
reasonably small. Specifically, the algorithm takes a constant bound $B \in
\reals$, and either finds the approximate value of the game, or verifies that
the value is strictly greater than $B$.

To formally define a Shapley game, we use $N$ to denote the number of states,
and $M$ to denote the number of actions. The game is defined by the following
two functions.
\begin{itemize}
	\item For each $s \le N$ and $j, k \le M$ the function $r(s, j, k)$ gives the
	reward at state $s$ when player one chooses action $j$ and player two 
	chooses action $k$.
	
	\item For each $s, s' \le N$ and $j, k \le M$ the function $p(s, s', j, k)$
	gives the probability of moving from state $s$ to state $s'$ when player one
	chooses action $j$ and player two chooses action $k$. It is required that
	$\sum_{s' = 1}^N p(s, s', j, k) = 1$ for all $s$, $j$, and $k$.
\end{itemize}

The game begins at a given starting state. In each round of
the game the players are at a state $s$, and play the matrix game at that 
state by picking an action from the set $\{1, 2, \dots, M\}$. 
The players are allowed to use
randomization to make this choice. Supposing that the first player chose action
$j$ and the second player chose the action $k$, the first player receives the
reward $r(s, j, k)$, and then a new state $s'$ is chosen according to the
probability distribution given by $p(s, \cdot, j, k)$. 

The reward in future rounds is \emph{discounted} by a factor of $\lambda$ where
$0 < \lambda < 1$ in each round. So if $r_1$, $r_2$, \dots is the infinite
sequence of rewards, the total reward paid by player two to player one is
$\sum_{i=1}^{\infty} \lambda^{i-1} \cdot r_i$, which, due to the choice of
$\lambda$, is always a finite value.

The two players play the game by specifying a probability distribution at each
state, which represents their strategy for playing at that state. Let $\Delta^M$
denote the $M$-dimensional simplex, which represents the strategy space for both
players at a single state. For each $x, y \in \Delta^M$, we overload notation by
defining the expected reward and next state functions.
\begin{align*}
r(s, x, y) &= \sum_{j=1}^M \sum_{k=1}^M x(j) \cdot y(k) \cdot r(s, i, j), \\
p(s, s', x, y) & = \sum_{j=1}^M \sum_{k=1}^M x(j) \cdot y(k) \cdot p(s, s', i, j).
\end{align*}

Shapley showed that these games are \emph{determined}~\cite{S53}, meaning that there is a
unique vector $v \in \reals^N$ such that $v_s$ is the \emph{value} of the game
starting at state $s$: player one has a strategy to ensure that the expected
reward is at least $v(s)$, while player two has a strategy to ensure that the
expected reward is at most $v(s)$. Furthermore, Shapley showed that this value
vector is the unique solution of the following \emph{optimality
	equations}~\cite{S53}. For each state $s$ we have the equation
\begin{align}
\label{eq:shapley}
v(s) = \min_{x \in \Delta^M} \max_{y \in \Delta^M} \left( r(s, x, y) +
\lambda \cdot \sum_{s'=1}^N  p(s, s', x, y) \cdot v_{s'} \right).
\end{align}
In other words, $v_s$ must be the value of the one-shot zero-sum game at $s$,
where the payoffs of this zero-sum game are determined by the values of the
other states given by $v_{s'}$. 

\begin{theorem}
	\label{thm:shapley}
	Let $\Gamma$ be a Shapley game with $N \in O(\sqrt[6]{\log M})$, unbounded number of 
	actions per state, and rewards in $[-c,c]$ for every state-action combination, where $c$ is a 
	constant. Furthermore, let $s$ be the starting state of the game. Let $B \in
	\reals$ be a constant. In quasi-polynomial
	time we can approximately compute the value of $\Gamma$  starting from $s$, if the value of
	every state is less than or equal to $B$, or decide that at least one of these
	values is greater than or equal to $B$.
\end{theorem}
\begin{proof}
	Let $v= (v(1),v(2),\ldots, v(N))$, and for every state $s$ let $x_s$ and $y_s$ 
	denote the strategy player one and player two choose at state $s$ respectively.
	%Consider Equation~\eqref{eq:shapley}. 
	Observe that $r(s,x_s,y_s)$ is an STM polynomial with variables $x$ and $y$ of 
	the form
	\begin{align*}
	\stm(A_{s1}, x_s,y_s) = \sum_{j=1}^M \sum_{k=1}^M  x_s(j) \cdot y_s(k) \cdot a_{s1}(j,k)
	\end{align*}
	where $a_{s1}(i,j,k) = r(s, j, k)$.
	
	Observe also that $\lambda \cdot \sum_{s'=1}^N  p(s, s', x_s, y_s) \cdot v_{s'}$ can 
	be written as an STM polynomial with variables $x,y$ and $v$ of the form 
	\begin{align*}
	\stm(A_{s2}, x,y,v) = \sum_{j=1}^M \sum_{k=1}^M \sum_{l=1}^N x_s(j) \cdot y_s(k) \cdot v(l) \cdot a_{s2}(j,k,l)
	\end{align*}
	where $a_{s2}(i,j,k) = \lambda \cdot p(s,l,j,k)$.
	
	Let us define $\tmv_s(x_s,y_s,v) = \stm(A_{s1}, x_s,y_s) +  \stm(A_{s2}, x_s,y_s,v)$;
	$\tmv_s(x_s,y_s,v)$ has length 2 and degree 1.
	
	Note that we can replace equation~\eqref{eq:shapley} with the following $2\cdot M$ 
	TMV polynomial constraints
	%\begin{align*}
	%\stm(A_{s1}, x,y) + & \stm(A_{s2}, x,y,v) - \big( \stm(A_{s1}, j,y) + \stm(A_{s2}, j,y,v)\big) \leq 0\\
	% & \qquad \text{for every action $j \leq M$ of player one}\\
	%\stm(A_{s1}, x,k) + & \stm(A_{s2}, x,k,v) - \big(\stm(A_{s1}, x,y) + \stm(A_{s2}, x,y,v) \big) \geq 0\\
	% & \qquad \text{for every action $k \leq M$ of player two.}
	%\end{align*}
	\begin{align*}
	\tmv(x_s,y_s,v) -\tmv(j,y_s,v) \leq 0 &\qquad \text{for every action $j \leq M$ of player one}\\
	\tmv(x_s,k,v) -\tmv(x_s,y_s,v) \geq 0 &\qquad \text{for every action $k \leq M$ of player two}.
	\end{align*}
	
	So, to approximate $v(s)$ it suffices to solve the \eetr problem defined by the 
	$2 \cdot M \cdot N$ constraints defined as above for every state $s \leq N$. 
	Observe, the \eetr problem has: $2N+1$ variables ($x_1$ through $x_N$, $y_1$ through
	$y_N$, and $v$); $2 \cdot M \cdot N$ TMV constraints; 
	$\gamma = \max\big\{1, \max_s v(s)\big\}$; 
	$\alpha = \max\big\{c, \lambda\cdot \max_{s,s',j,k}p(s,s',j,k)\big\} = \max\{c ,1\}$,
	since $\lambda < 1$ and $\max_{s,s',j,k}p(s,s',j,k)<1$. 
	So, if $N \in O( \sqrt[6]{\log M} )$, $\max_s v(s)$ is constant, and $c$ is a constant, we can use
	Theorem~\ref{thm:main} and derive a QPTAS for~\eqref{eq:shapley}.
	
	Finally, we note that an approximate solution to~\eqref{eq:shapley} gives an
	approximation of the value vector itself. This is because Shapley has shown
	that, when $v$ is treated as a variable, the optimality equation given
	in~\eqref{eq:shapley} is a \emph{contraction map}. The value vector is a fixed
	point of this contraction map, and the uniqueness of the
	value vector is guaranteed by Banach's fixed point theorem. Our algorithm
	produces an approximate fixed point of the optimality equations. It is easy to
	show, using the contraction map property, that an approximate fixed point must
	be close to an exact fixed point.
\end{proof}

\subsection{Approximate consensus halving}\label{sec:QPTAS}

In this section we show that an approximate solution to the consensus
halving problem can be found in quasi-polynomial time when each agent's 
valuation function is a single polynomial of constant or even polylogarithmic degree.
We will do so by formulating the problem as a constrained \eetr instance, and then applying Theorem \ref{thm:main}. 

This result first appeared in \cite{I-DFMS19,DFMS19} and implies that these instances can be solved approximately using a
polylogarithmic number of cuts. We note that this is one of the most general classes
of instances for which we could hope to prove such a result: any instance in
which $n$ agents desire completely disjoint portions of the object can only be
solved by an $n$-cut, and piecewise linear functions are capable of producing such a situation.
So in a sense, we are exploiting the fact that this situation cannot arise when
the agents have non-piecewise polynomial valuation functions. 
%The proof of the following lemma can be found in Appendix \ref{app:qptas}.

%In this section we show how to get a quasi-PTAS for \CH in a restricted class of
%``nice'' valuation functions for the agents. 

%\begin{definition}
%\label{def:nice}
%The integral $F_i$ of the valuation function of agent $i$ is {\em nice} if there
%exist $j \in O(\log n)$ and $a_0, a_1, \ldots, a_l \in O(\log n)$ such that $F_i$ 
%can be written as $F_i(x) = \sum_{j=0}^l a_jx^j$ for every $x \in [0,1]$. 
%\end{definition}

%More specifically, we show that we can formulate the consensus halving problem as 
%an \eetr problem. Then, since $F_i$s are nice, we use the recent result of~\cite{DFMS18} 
%to prove that there exists an $O(\log n)$-cut and pieces $A_+$ and $A_-$ such that:
%\begin{itemize}
%\item where every cut point is a multiple of $1/k = \frac{\eps^2}{O(\log n)}$;
%\item $|F_i(A_+) -F_i(A_-)| \leq \eps$, for every agent $i$.
%\end{itemize} 

\begin{lemma}\label{lem:qptas}
	%\CH admits a QPTAS when the valuation function of every agent is nice.
	For every \CH instance with $n$ agents, and every $\eps > 0$, if each agent's valuation function $F_i$
	is a single polynomial of degree at most $O(poly \log n)$, then 
	there exists a $k$-cut, where $k:=O(poly \log n)/\eps^5$, and pieces $A_+$ and $A_-$ such that:
	\begin{itemize}
		\item every cut point is a multiple of $1/k = \frac{\eps^5}{O(poly \log n)}$;
		\item $|F_i(A_+) -F_i(A_-)| \leq \eps$, for every agent $i$.
	\end{itemize} 
\end{lemma}

\begin{proof}
	Since each agent $i$ has a polynomial valuation function, there is a $d \in O(\log
	n)$ and constants $a_0, a_1, \ldots, a_l$ such that each function $F_i$ can be
	written as $F_i(t) = \sum_{j=0}^d a_j \cdot t^j$. 
	
	To prove the theorem, we will formulate the problem as a constrained \eetr instance, and apply Theorem \ref{thm:main}, which proves the claim.
	We first write a simple \etr formula for consensus halving with polynomial
	valuation functions. If a consensus halving instance has a solution, then it
	also has one in which the cuts are \emph{strictly alternating}, meaning that
	\begin{align*}
	F_i(A_+) &= \sum_{j=1}^{\lfloor n/2 \rfloor} \big(F_i(t_{2j}) -
	F_i(t_{2j-1})\big),\\
	F_i(A_-) &= \sum_{j=1}^{\lceil n/2 \rceil} \big(F_i(t_{2j-1}) -
	F_i(t_{2j-2})\big),
	\end{align*}
	where the cut is the tuple $(t_1, t_2, \dots, t_n)$, with $0 \leq t_1 \leq \dots \leq t_n \leq 1$ and $t_{0} := 0$, $t_{n+1}:=1$.
	
	In this encoding, we have no need to encode which set a particular cut belongs
	to, and so we can encode a $n$-cut as an element of the $n$-simplex
	$x:=(x_1, x_2, \dots, x_{n+1}) \in \Delta^{n+1}$, where $x_{i}:= t_{i} - t_{i-1}$.
	%%%%%%%%%%%%%%%%%%%%%%%
	From the latter, it is easy to see that 
	\begin{align}\label{eq:t_i}
	t_{i}:=\sum_{j=1}^{i} x_{j}.
	\end{align}
	For $j \in \{0,1, \dots, n \}$, let us denote by $1^j$ and $0^j$ a $j$-tuple of $1$'s and $0$'s respectively. Let us also define the $n$-dimensional vector $v_{j} := (0^{j}, 1^{n-j})$. Now observe that any $n$-cut $t := (t_1, t_2, \dots, t_n)$ can be represented by a $n$-dimensional point which is in fact a convex combination of the $n+1$ vectors $v_{j}$, $j \in \{ 0, 1, \dots, n \}$. In particular, from \eqref{eq:t_i} it is easy to see that 
	\begin{align*}
	t := (t_1, t_2, \dots, t_n) = \sum_{j=1}^{n+1} x_{j} \cdot v_{j-1}.
	\end{align*}
	%%%%%%%%%%%%%%%%%%%%%%%
	
	Hence, we can encode the problem as an \etr formula
	\begin{equation*}
	\exists t \cdot \left( \bigwedge\limits_{i=1}^n F_i(A_+) = F_i(A_-) \right)
	\land t \in C,
	\end{equation*}
	where $C$ is the convex hull of the vectors $v_{0}, v_{1}, \dots, v_{n}$.
	This formula has $n$ constraints, one for each agent, and a single constraint
	bounding the variables in the convex set $C$ which can be expressed by $n+1$ vectors, namely $v_{j}$, $j \in \{ 0,1, \dots, n \}$.

	The main theorem of \cite{W-DFMS18} allows us to leave the constraint
	$t \in C$ unchanged, but insists that we weaken the others.
	Specifically each constraint is weakened so that 
	only $F_i(A_+)- F_i(A_-)\leq \eps$ and $F_i(A_+)- F_i(A_-)\geq -\eps$ are
	enforced, which implies that 
	$|F_i(A_+)- F_i(A_-)| \leq \eps$. This is sufficient to encode an approximate
	solution to the problem.
	%\begin{align*}
	%%\label{eq:qptas1}
	%|\sum_{j=1}^{\lfloor \frac{n-1}{2} \rfloor} \big(F_i(x_{2j}) - F_i(x_{2j-1})\big)- \sum_{j=1}^{\lfloor \frac{n-1}{2} \rfloor} \big(F_i(x_{2j-1}) - F_i(x_{2j})\big)| & \leq \eps\\
	%%%\label{eq:qptas2}
	%%\sum_{j=1}^{\lfloor \frac{n-1}{2} \rfloor} \big(F_i(x_{2j}) - F_i(x_{2j-1})\big)- \sum_{j=1}^{\lfloor \frac{n-1}{2} \rfloor} \big(F_i(x_{2j-1}) - F_i(x_{2j})\big) & \geq -\eps.
	%\end{align*}
	
	The constructed \eetr instance has one vector-variable $t \in C$ 
	and $2n$ constraints. Let us now study one of the constraints of the \eetr instance.
	\begin{align*}
	\sum_{j=1}^{\lfloor n/2 \rfloor} \big(F_i(t_{2j}) - F_i(t_{2j-1})\big)- \sum_{j=1}^{\lceil n/2 \rceil} \big(F_i(t_{2j-1}) - F_i(t_{2j-2})\big) & \leq \eps.
	\end{align*}
	Using the representation of $F_i$, we can write down a constraint as 
	$\sum_{k=0}^d a_k\cdot h_k(t_1, t_2, \ldots, t_n) \leq \eps$, where 
	$h_{k}(t_1, t_2, \ldots, t_n)$ is a sum of monomials, each one of degree $d$. $F_{i}$ depends on $t_{0}$ and $t_{n+1}$ as well, but recall that these are $0$ and $1$ respectively. 
	
	The term $a_k\cdot h_k(t_1, t_2, \ldots, t_n)$ is a simple tensor multivariate polynomial
	with one variable of degree $k$, which we will denote by $STM(H_k,t^k)$. Under this notation
	$H_k$ is a $k$-dimensional tensor where vector $t$ is applied $k$ times. Hence,
	every constraint is a sum of $d+1$ simple tensor multivariate polynomials, i.e. a TMV polynomial of maximum
	degree $d$ constructed by $d+1$ STM polynomials. Furthermore, $||v_{j}||_{\infty} \leq 1$ for all $j \in \{ 0, 1, \dots, n \}$, and for every constraint, the maximum absolute coefficient is
	constant by definition, and the degree $d$ is $O(poly \log n)$. Hence, we can apply Theorem \ref{thm:main} and get
	the claimed result.   
\end{proof}

As a consequence, we can perform a brute force search over all possible $k$-cuts to find an approximate solution, which can be carried out in 
$n^{O(poly \log n/\eps^5)}$ time.

\begin{theorem}
	\CH admits a QPTAS when the valuation function of every agent 
	is a single polynomial of degree $O(poly \log n)$.
\end{theorem}

\subsection{Optimization problems}

Our framework can provide approximation schemes for optimization problems with one 
vector variable $x \in \reals^p$ with polynomial constraints over bounded convex sets.
Formally,
\begin{align*}
\max~ & ~h(x)\\
s.t.~ &~ h_1(x) \geq 0, \ldots, h_m(x) \geq 0  \\
&~ x \in \conv(c_1, \ldots, c_l)
\end{align*}
where $h(x), h_1(x), \ldots, h_m(x)$ are polynomials with respect to vector $x$; 
for example $h(x) = x^TAx$, where $A$ is an $p \times p$ matrix, 
subject to $h_1(x) = x^Tx - \frac{1}{10} \geq 0$ and $x \in \Delta^p$.  
We will call the polynomials $h_i$ {\em solution-constraints}.
Optimization problems of this kind  received a lot of attention over the years~\cite{K08,de2006ptas,de2015alternative,F04}. 

For optimization problems, we sample from the solution that achieves the
maximum when we apply Theorem~\ref{thm:main}, in order to prove that there is a
$k$-uniform solution that is close to the maximum. Our algorithm
enumerates all $k$-uniform profiles, and outputs the
one that maximizes the objective function. Using this technique,
Theorem \ref{thm:main} gives the following results.
\begin{enumerate}
	\item 
	\label{opt:1}
	There is a PTAS if $h(x)$ is a STM polynomial of maximum degree independent of $p$, the number of
	solution-constraints is independent of $p$, and $l = \poly(p)$.	
	\item 
	\label{opt:2}
	There is QPTAS if $h(x)$ is a STM polynomial of maximum degree up to $\polylog p$, the 
	number of solution-constraints is $\poly(p)$, and $l = \poly(p)$.	
\end{enumerate}

To the best of our knowledge, the second result is new.
The first result was already known, however it was proven using
completely different techniques: in~\cite{BK02} it was proven for the
special case of degree two, in~\cite{F04} it was extended to any fixed degree,
and alternative proofs of the fixed degree case were also given in
\cite{de2006ptas,de2015alternative}. We highlight that in all of the
aforementioned results solution constraints were not allowed. Note that
unless \NP= \texttt{ZPP} there is no FPTAS for quadratic programming even when the
variables are constrained in the simplex~\cite{K08}. Hence, our results can be
seen as a partial answer to the important question posed in \cite{K08}: \textit{What is a complete classification of functions that allow a PTAS?}

Furthermore, as shown in Theorem \ref{thm:multi-obj} this technique yields a generalized algorithm for multi-objective optimization problems which, to the best of our knowledge, is a completely new result.

\subsection{Tensor problems}

Our framework provides quasi-polynomial time algorithms for deciding the
existence of approximate eigenvalues and approximate eigenvectors of tensors in
$\reals^{p \times p \times p}$, where the elements are bounded by a constant,
where the solutions are required to be in a bounded convex set.  In~\cite{HL13} it is
proven that there is no PTAS for these problems when the domain is unrestricted.
To the best of our knowledge this is the first positive result for the problem
even in this, restricted, setting.  

\begin{definition}
	\label{def:eigen}
	The nonzero vector $x \in \reals^p$ is an {\em eigenvector} of tensor
	$A \in \reals^{p \times p \times p}$ if there exists an {\em eigenvalue} 
	$\lambda \in \reals$ such that for every $k \in [p]$ it holds that
	\begin{align}
	\label{eq:eigen}
	\sum_i^n \sum_j^n a(i,j,k) \cdot x(i) \cdot x(j) = \lambda \cdot x(k).
	\end{align}
\end{definition}

\begin{theorem}
	\label{thm:eigen}
	Let $A$ be an $\reals^{p \times p \times p}$ tensor with entries in $[-c,c]$, where $c$ is a 
	constant. Furthermore, let $B \in \reals$ be a constant and let $\ycal$ be a bounded convex 
	set where $\|\ycal\|_\infty$ is a constant. In a quasi-polynomial time we can compute an 
	eigenvalue-eigenvector pair $(\lambda, x)$ that approximately satisfy~\eqref{eq:eigen} 
	such that $\lambda \leq B$ and $x \in \ycal$, or decide that no such pair exists.
\end{theorem}
\begin{proof}
	Observe that $\sum_i^n \sum_j^n a(i,j,k) \cdot x(i) \cdot x(j)$ can be written as an STM
	polynomial $\stm(A_1,x^2)$ where $a_1(i,j) = a(i,j,k)$. Furthermore, let $\ell$ be a 
	$p$-dimensional vector. Then,  $\lambda \cdot x(k)$ can be written as an STM polynomial 
	$\stm(A_2,x,\ell)$, where $a_2(k,1) = 1$ and zero otherwise.
	
	So, Equation \eqref{eq:eigen} can be written as an TMV polynomial constraint of degree 2 and 
	length 2, with two vector variables, $x$ and $\ell$. So, the problem of computing an 
	eigenvalue-eigenvector pair that approximately satisfy \eqref{eq:eigen} can be written as an 
	\eetr problem with $p$ TMV polynomial constraints of degree 2 and length 2 and two vector
	variables. Hence, we can use Theorem~\ref{thm:main} with $\gamma=\|\ycal\|_\infty$  which 
	is constant, $\alpha = c$, $n=2$, $t=2$, $d=2$, and $m = p$ to find a solution if exists, 
	or decide that no such solution exists.
\end{proof}

\subsection{Computational geometry}

Finally, we note that our theorem can be applied to problems in computational
geometry, although the results are not as general as one may hope. Many problems
in this field are known to be \etr-complete, including, for example, the Steinitz
problem for 4-polytopes, inscribed polytopes and Delaunay triangulations,
polyhedral complexes, segment intersection graphs, disk intersection graphs, dot
product graphs, linkages, unit distance graphs, point visibility graphs,
rectilinear crossing number, and simultaneous graph embeddings. We refer the
reader to the survey of Cardinal \cite{C15} for further details.

All of these problems can be formulated in $\eps$-\etr, and indeed our theorem
does give results for these problems. However, our requirement that the bounding
convex set be given explicitly limits their applicability. Most computational
geometry problems are naturally constrained by  a cube, so while 
Corollary~\ref{cor:nptas} does give \NP algorithms, we do not get QPTASs unless we further
restrict the convex set. Here we formulate QPTASs for
the segment intersection graph and the unit disk intersection graph
problems when the solutions are restricted to lie in a simplex. While it is not
clear that either problem has natural applications that are restricted in this
way, we do think that future work may be able to derive sampling theorems that
are more tailored towards the computational geometry setting.

\subsubsection{Segment intersection graphs}

\paragraph*{Definitions} Let $G$ be an undirected graph with vertex set $\{v_1, v_2, \dots, v_n \}$. We say that $G$ is a \textit{segment graph} if there are straight segments $s_1, s_2, \dots, s_n$ in the plane such that, for every $i,j, 1 \leq i < j \leq n$, the segments $s_i$ and $s_j$ have a common point if and only if $\{v_i, v_j\} \in E(G)$.
% We let SEG denote the class of all segment graphs.
%
%We first give a description for the problem with $\eps = 0$ and then we generalize for arbitrary $\eps \geq 0$. 

By a suitable rotation of the co-ordinate system we can achieve that none of the segments is vertical. Then the segment $s_i$ representing vertex $v_i$ can be algebraically described as the set $\{(x,y) \in \reals^2 : y = a_{i} x + b_{i} , c_{i} \leq x \leq d_{i} \}$ for some real numbers $a_{i}, b_{i}, c_{i}, d_{i}$. We say that $G$ is a \textit{simplex K segment graph} if the real numbers $a_{i}, b_{i}, c_{i}, d_{i}$, $i=1,2, \dots n$ are under the constraints 
\begin{align*}
& a_{i}, b_{i}, c_{i}, d_{i} \geq 0, \text{ for every } i=1,2, \dots n, \text{ and} \\
& \sum_{i=1}^{n} (a_{i} + b_{i} + c_{i} + d_{i}) = K, \quad \text{where } K > 0 \text{ is a given constant.}
\end{align*}
We let SIM-K-SEG denote the class of all simplex $K$ segment graphs with parameter $K > 0$.

The problem $\eps$-RECOG(SIM-K-SEG) is defined as follows. Given an abstract undirected graph $G$, does it belong with tolerance $\eps$ to SIM-K-SEG?

\paragraph*{Formulation of $\eps$-RECOG(SIM-K-SEG)}
We first give a description for the problem with $\eps = 0$ and then we generalize for arbitrary $\eps \geq 0$. The formulation is taken from \cite{Matousek14}.

Letting $l_{i}$ be the line containing $s_{i}$, we note that $s_{i} \cap s_{j} \neq \emptyset$ if $l_{i}$ and $l_{j}$ intersect in a single point whose $x$-coordinate lies in both the intervals $[c_{i},d_{i}]$ and $[c_{j},d_{j}]$. It is easy to see that the $x$-coordinate equals $\frac{b_{j} - b_{i}}{a_{i} - a_{j}}$. 

Now we turn to the general case where $\eps \geq 0$. Let us introduce variables $A_{i}, B_{i}, C_{i}, D_{i}$ representing the unknown quantities $a_{i}, b_{i}, c_{i}, d_{i}$, $i = 1,2, \dots, n$. By the problem's definition we require the vector $(A_{1}, B_{1}, C_{1}, D_{1}, \dots , A_{n}, B_{n}, C_{n}, D_{n})$ to be in the ($4n-1$)-simplex with parameter $K$.
Then $s_{i} \cap s_{j} \neq \emptyset$ can be expressed by the following predicate:
\begin{align*}
\text{INTS}(A_{i}, B_{i}, C_{i}, D_{i},& A_{j}, B_{j}, C_{j}, D_{j}) = \\
%& ( A_{i} \eeps A_{j} \land B_{i} \eeps B_{j} \land \lnot (D_{i} \leps C_{j} \lor D_{j} \leps C_{i} ) ) \\	
( & A_{i} \geps A_{j} \land C_{i}(A_{i} - A_{j}) \leeps B_{j} - B_{i} \leeps D_{i}(A_{i} - A_{j}) \\
& \land C_{j}(A_{i} - A_{j}) \leeps B_{j} - B_{i} \leeps D_{j}(A_{i} - A_{j}) ) \\
\lor ( & A_{i} \leps A_{j} \land C_{i}(A_{i} - A_{j}) \geeps B_{j} - B_{i} \geeps D_{i}(A_{i} - A_{j}) \\
& \land C_{j}(A_{i} - A_{j}) \geeps B_{j} - B_{i} \geeps D_{j}(A_{i} - A_{j}) )
\end{align*}
(this is only correct if we ``globally'' assume that $C_{i} \leeps D_{i}$ for all $i$). The existence of a SEG-representation of $G$ can then be expressed by the formula 
\begin{align*}
( \exists A_{1}B_{1}C_{1}&D_{1} \dots A_{n}B_{n}C_{n}D_{n}K) \left( \bigwedge_{i=1}^{n} C_i \leeps D_i \right) \\
& \wedge \left( \bigwedge_{\{i,j\} \in E} \text{INTS}(A_{i}, B_{i}, C_{i}, D_{i}, A_{j}, B_{j}, C_{j}, D_{j}) \right) \\
& \wedge \left( \bigwedge_{\{i,j\} \notin E} \lnot \text{INTS}(A_{i}, B_{i}, C_{i}, D_{i}, A_{j}, B_{j}, C_{j}, D_{j}) \right) 
%\\
%& \wedge \left( \bigwedge_{i=1}^n A_i \geq 0 \right) \wedge \left( \bigwedge_{i=1}^n B_i \geq 0 \right) \wedge \left( \bigwedge_{i=1}^n C_i \geq 0 \right) \wedge \left( \bigwedge_{i=1}^n D_i \geq 0 \right) \\
%& \wedge \sum_{i=1}^n (A_i + B_i + C_i + D_i) = K > 0 ,
\end{align*}
%
%where the last two rows are the constraints that stem from the inclusion in the simplex.

\begin{theorem}
	There is an algorithm
	that runs in time $n^{O(K^2 \cdot \log n / \eps^2)}$ and either finds
	a vector $(A_{1},B_{1},C_{1},D_{1}, \dots ,A_{n},B_{n},C_{n},D_{n})$ that is a solution to $\eps$-RECOG(SIM-K-SEG), or determines that there is no solution to
	$0$-RECOG(SIM-K-SEG). 
\end{theorem}
\begin{proof}
	We set $x = (A_{1},B_{1},C_{1},D_{1}, \dots ,A_{n},B_{n},C_{n},D_{n})$ and
	$F(x)$ to be the above formula that we constructed. Their combination makes an
	$\eps$-\etr instance. Vector $x$ is constrained over the convex hull defined by
	the vertices of the $(4n-1)$-simplex, i.e. vectors $v_{i} \in \reals^{4n}$, $i
	\in \{1,2, \dots 4n \}$ with their $i$-th element equal to $K$ and the rest
	equal to 0. Therefore the cardinality of our convex set is $m = 4n$, and $\gamma
	= K$. By looking at the formula we can conclude that $a = 1$, $t = 4$, and $d =
	2$. By Theorem~\ref{thm:kunif} the result follows.
\end{proof}

\subsubsection{Unit disk intersection graphs}

\paragraph*{Definitions} Let $G$ be an undirected graph with vertex set $\{v_1, v_2, \dots, v_n \}$. We say that $G$ is a \textit{unit disk intersection graph} or \textit{unit disk graph} if there are disks $d_1, d_2, \dots, d_n$ (in the plane) with radius 1 such that, for every $i,j, 1 \leq i < j \leq n$, the disks $d_i$ and $d_j$ have more than one points common (i.e. an area) if and only if $\{v_i, v_j\} \in E(G)$.
% We let SEG denote the class of all segment graphs.
%
%We first give a description for the problem with $\eps = 0$ and then we generalize for arbitrary $\eps \geq 0$. 

The disk $d_i$ representing vertex $v_i$ can be algebraically described as the set $\{(x,y) \in \reals^2 : (x - x_{i})^2 + (y - y_i)^2 \leq 1 \}$ for some real numbers $x_{i}, y_{i}$ that determine the centre of the disk. We say that $G$ is a \textit{simplex K unit disk graph} if the real numbers $x_{i}, y_{i}$, $i=1,2, \dots n$ are under the constraints 
\begin{align*}
& x_{i}, y_{i} \geq 0, \text{ for every } i=1,2, \dots n, \text{ and} \\
& \sum_{i=1}^{n} (x_{i} + y_{i}) = K, \quad \text{where } K > 0 \text{ is a given constant.}
\end{align*}
We let SIM-K-UDG denote the class of all simplex $K$ unit disk graphs with parameter $K > 0$.

The problem $\eps$-RECOG(SIM-K-UDG) is defined as follows. Given an abstract undirected graph $G$, does it belong with tolerance $\eps$ to SIM-K-UDG?

\paragraph*{Formulation of $\eps$-RECOG(SIM-K-UDG)}

Let us introduce variables $X_{i}, Y_{i}$ representing the unknown quantities $x_{i}, y_{i}$, $i = 1,2, \dots, n$. We require the vector $(X_{1}, Y_{1}, \dots , X_{n}, Y_{n})$ to be in the ($2n-1$)-simplex with parameter $K$.
Then we consider an $\eps$-intersection $d_{i} \cap_{\eps} d_{j} \neq \emptyset$ to happen if:
\begin{align*}
\sqrt{(X_{i} - X_{j})^2 + (Y_{i} - Y_{j})^2} < 2 + \eps
\end{align*}
and an $\eps$-non-intersection $d_{i} \cap_{\eps} d_{j} = \emptyset$ to happen if:
\begin{align*}
\sqrt{(X_{i} - X_{j})^2 + (Y_{i} - Y_{j})^2} \geq 2 - \eps
\end{align*}
The existence of a UDG-representation of $G$ can then be expressed by the formula 
\begin{align*}
( \exists X_{1}Y_{1}& \dots X_{n}Y_{n})  \\
& \left( \bigwedge_{\{i,j\} \in E} (X_{i} - X_{j}) \cdot (X_{i} - X_{j}) + (Y_{i} - Y_{j}) \cdot (Y_{i} - Y_{j}) < 4 + 2 \eps + \eps^2 \right) \\
\wedge & \left( \bigwedge_{\{i,j\} \notin E} (X_{i} - X_{j}) \cdot (X_{i} - X_{j}) + (Y_{i} - Y_{j}) \cdot (Y_{i} - Y_{j}) \geq 4 - 2 \eps + \eps^2 \right) 
\end{align*}
%
%where the last two rows are the constraints that stem from the inclusion in the simplex.

\begin{theorem}
	There is an algorithm
	that runs in time $n^{O(K^2 \cdot \log n / \eps^2)}$ and either finds
	a vector $(X_{1},Y_{1}, \dots ,X_{n},Y_{n})$ that is a solution to $\eps$-RECOG(SIM-K-UDG), or determines that there is no solution to 0-RECOG(SIM-K-UDG). 
\end{theorem}
\begin{proof}
	We set $x = (X_{1},Y_{1}, \dots ,X_{n},Y_{n})$ and $F(x)$ to be the above
	formula that we constructed. Their combination makes an $\eps$-\etr instance.
	Vector $x$ is constrained over the convex set defined by the vertices of the
	$(2n-1)$-simplex, i.e. vectors $v_{i} \in \reals^{2n}$, $i \in \{1,2, \dots 2n
	\}$ with their $i$-th element equal to $K$ and the rest equal to 0. Therefore
	the cardinality of our convex set is $m = 2n$, and $\gamma = K$. By looking at
	the formula we can conclude that $a = 2$, $t = 7$, and $d = 2$. By
	Theorem~\ref{thm:kunif} the result follows.
\end{proof}

\section{Discussion and Open Problems}

It seems that \etr is a class which captures decision problems that are a lot harder than these in \NP (under standard complexity assumptions) because either they do not have truth certificates of polynomial length or because the certificate cannot be checked in polynomial time. One can think of \etr and thus \texttt{Function} \etr (\fetr) and \texttt{Total Function} \etr (\tfetr) as being the analogues of \NP, \FNP and \TFNP respectively in the Blum-Shub-Smale (BSS) model of computation \cite{BSS89}, in which computing functions over real numbers is as costly as is computing functions over rational numbers in Turing machines.

In this paper we provide a general framework for approximation schemes, a framework designed for problems in a subclass of \etr (or more precisely, \fetr). In particular, since some function problems in \TFNP or, in general, \FNP (whose corresponding decision problems are in \NP), have polynomial or quasi-polynomial time approximation schemes (PTAS/QPTAS), we study harder problems in \tfetr or \fetr, and seek similar approximation schemes. In a beautiful turn of events, we show that PTASs and QPTASs exist for a wide class of problems in \fetr. By extending the well-known Lipton-Markakis-Mehta (LMM) technique that yields the best possible algorithm (under standard complexity assumptions) for computing approximate Nash equilibria in strategic games, we provide a general framework that gives in a standardized way, approximation algorithms of the same quality as the state of the art for some problems, while for some other problems these algorithms are the first to achieve an efficient approximation. Interestingly, approximation techniques that work inside \FNP, transcend it, and reach \fetr.

For a given constrained \eetr instance whose variables' domain is the convex hull of $l$ vectors, we presented an algorithm which runs in time $l^{O(k)}$, for $k$ indicated in Theorem \ref{thm:main}, that either computes a solution or responds that a solution to the exact instance does not exist. This algorithm is a QPTAS or PTAS for many well-known problems. However, our algorithm, being an extension of the LMM algorithm, for some problems does not have better running time than the state of the art algorithms that are tailored to these problems. The most important open problem is to make the quantity $k$ depend logarithmically on crucial parameters, such as the number of variables $n$ and the degree of the polynomials $d$, instead of polynomially. This would generalize many algorithms, such as the PTAS for computing an $\eps$-Nash equilibrium in anonymous games \cite{DP09} and the best algorithm for computing an $\eps$-Nash equilibrium in general multi-player normal form games \cite{BBP}.

\section*{Acknowledgements}
P. Spirakis wishes to dedicate this paper to the memory of his late father in law Mathematician and Professor Dimitrios Chrysofakis, who was among the first in Greece to work on tensor analysis.

\bibliographystyle{abbrv}
\bibliography{references}

\end{document}